\documentclass[12pt,a4paper]{article}
\usepackage{amsmath,amssymb,amsthm,graphicx,color}
\newtheorem{theorem}{Theorem}
\newtheorem{proposition}[theorem]{Proposition}
\newtheorem{lemma}[theorem]{Lemma}
\newtheorem{corollary}[theorem]{Corollary}
\newtheorem{definition}[theorem]{Definition}
\newtheorem{assumption}[theorem]{Assumption}

\def\supp{\mathop{\rm supp}\nolimits}
\def\dist{\mathop{\rm dist}\nolimits}
\def\Ker{\mathop{\rm Ker}\nolimits}

\def\Im{\mathop{\rm Im}\nolimits}
\def\min{\mathop{\rm min}\nolimits}
\def\max{\mathop{\rm max}\nolimits}
\begin{document}
\begin{center}
 {\large\bf 
A self-adjointness criterion for 
the Schr\"odinger operator with infinitely many point interactions
and its application to random operators
}

by

Masahiro Kaminaga\footnote{
Department of Information Technology,
Tohoku Gakuin University,
Tagajo 985-8537, Japan.\ 
E-mail: kaminaga@mail.tohoku-gakuin.ac.jp
}, %
Takuya Mine\footnote{
Faculty of Arts and Sciences,
Kyoto Institute of Technology,
Matsugasaki, Sakyo-ku,
Kyoto 606-8585, Japan.\ 
E-mail: mine@kit.ac.jp
}, %
and Fumihiko Nakano\footnote{
Department of Mathematics,
Gakushuin University,
Mejiro 1-1-5,
Toshima-ku,
Tokyo 171-0031, Japan.\ 
E-mail: fumihiko@math.gakushuin.ac.jp
}
\end{center}
\begin{abstract}
We prove the Schr\"odinger operator
with infinitely many point interactions in $\mathbb{R}^d$ $(d=1,2,3)$
is self-adjoint if the support of the interactions
is decomposed into uniformly discrete clusters.
Using this fact,
we prove the self-adjointness of 
the Schr\"odinger operator with 
point interactions on a random perturbation of a lattice
or on the Poisson configuration.
We also determine the spectrum of the Schr\"odinger operators
with random point interactions of Poisson--Anderson type.
\end{abstract}

\section{Introduction}
Let $\Gamma$ be a 
discrete subset of $\mathbb{R}^d$ ($d=1,2,3$)
which is \textit{locally finite},
that is,
$\#(\Gamma\cap K)<\infty$ for any compact subset $K$ of 
$\mathbb{R}^d$,
where the symbol $\#S$ is the cardinality of a set $S$.
We define the \textit{minimal operator}
$H_{\Gamma,\min}$ by
\begin{align*}
  H_{\Gamma,\min} u = -\Delta u,\quad 
D(H_{\Gamma,\min})= C_0^\infty(\mathbb{R}^d\setminus \Gamma),
\end{align*}
where 
$\Delta=\sum_{j=1}^d \partial^2/\partial x_j^2$ is the Laplace operator.
Clearly $H_{\Gamma,\min}$ is a densely defined symmetric operator,
and it is well-known that
the deficiency indices $n_\pm(H_{\Gamma,\min})$ are given by
\begin{equation*}
n_\pm(H_{\Gamma,\min}):=\dim \mathcal{K}_{\Gamma,\pm}=
\begin{cases}
2\#\Gamma & (d=1),\cr
\#\Gamma & (d=2,3),
\end{cases}
\end{equation*}
where  $\mathcal{K}_{\Gamma,\pm}:=\Ker ({H_{\Gamma, \min}}^*\mp i)$
are the deficiency subspaces (see e.g.\ \cite{Al-Ge-Ho-Ho}).
So $H_{\Gamma,\min}$ is not essentially self-adjoint
unless $\Gamma=\emptyset$.
A self-adjoint extension of $H_{\Gamma,\min}$ is
called the Schr\"odinger operator with \textit{point interactions},
since the support of the interactions is concentrated on countable number of
points in $\mathbb{R}^d$.
The Schr\"odinger operator with point interactions
is known as a typical example of solvable models in quantum mechanics,
and numerous works are devoted to the study of this model 
or its perturbation by a scalar potential or a magnetic vector potential.
The book \cite{Al-Ge-Ho-Ho} 
contains most of fundamental facts about this subject 
and exhaustive list of references up to 2004.
The papers \cite{Br-Ge-Pa,Ko-Ma} also give us recent development of 
this subject.

There are mainly three popular methods of defining
self-adjoint extensions $H$ of $H_{\Gamma,\min}$.
Here we denote the free Laplacian by $H_0$, that is,
$H_0=-\Delta$ with $D(H_0)=H^2(\mathbb{R}^d)$.
\begin{enumerate}
 \item Calculate the deficiency subspaces $\mathcal{K}_{\Gamma,\pm}$,
and give the difference of 
the resolvent operators $(H-z)^{-1}-(H_0-z)^{-1}$ ($\Im z\not=0$)
for the desired self-adjoint extension $H$
by using von Neumann's theory and Krein's resolvent formula.

 \item Introduce a scalar potential $V$,
choose the renormalizatoin factor $\lambda(\epsilon)$ appropriately,
and define the operator $H $ as the norm resolvent limit
\begin{align}
\label{intro00}
H= \lim_{\epsilon\to 0}H_\epsilon,\quad
H_\epsilon=-\Delta+\lambda(\epsilon) \epsilon^{-d}V(\cdot/\epsilon). 
\end{align}

 \item 
Define the operator domain $D(H)$ of
the desired self-adjoint extension $H$
in terms of the boundary conditions at $\gamma\in\Gamma$.
\end{enumerate}
These methods are mutually related with each other,
and give the same operators consequently.
Historically,
the seminal works by Kronig--Penney \cite{Kr-Pe} ($d=1$)
and Thomas \cite{Th} ($d=3$) start from the method (ii),
and conclude the limiting operators are described by the method (iii).
Bethe--Peierls \cite{Be-Pi} also obtain a similar boundary condition 
for $d=3$.
Berezin--Faddeev \cite{Be-Fa} start from the method (ii) for $d=3$
by using the cut-off in the momentum space,
and show the limiting operator is also defined by the method (i).
After the paper \cite{Be-Fa},
 the method (i) becomes probably the most commonly used one.
It is mathematically rigorous
and useful in the analysis of spectral and scattering properties
of the system,
since various quantities 
(e.g.\ spectrum, scattering amplitude, resonance, etc.)
are defined via the resolvent operator.
The characteristic feature in the method (ii) is the dependence of
the renormalization factor $\lambda(\epsilon)$ on the dimension $d$.
We can take $\lambda(\epsilon)=1$ for $d=1$,
but $\lambda(\epsilon)\to0$ as $\epsilon\to 0$ for $d=2,3$.
Recently, the method (iii) is reformulated  
in terms of the \textit{boundary triplet}
(see \cite{Br-Ge-Pa,Ko-Ma} and references therein).
The method (iii) is useful
when we cannot calculate the deficiency subspaces explicitly,
e.g.\ the point interactions on a Riemannian manifold, etc.
In the present paper we adapt the method (iii),
as explained below.

We define the \textit{maximal operator}
$H_{\Gamma,\max}$
by $H_{\Gamma,\max}={H_{\Gamma,\min}}^*$,
the adjoint operator of $H_{\Gamma,\min}$.
The operator $H_{\Gamma,\max}$ is explicitly given by 
\begin{align*}
  H_{\Gamma,\max} u = -\Delta u,\quad 
D(H_{\Gamma,\max})= 
\{u \in L^2(\mathbb{R}^d) ;\ \Delta u \in L^2(\mathbb{R}^d)\},
\end{align*}
where the Laplacian $\Delta$ 
is regarded as a linear operator on 
the space of Schwartz distributions
${\cal D}'(\mathbb{R}^d\setminus\Gamma)$
on $\mathbb{R}^d\setminus\Gamma$
(see \cite{Al-Ge-Ho-Ho} or Proposition \ref{proposition_hmax} below;
we interpret $\Delta$ in this sense in the sequel).
When $d=1$, an element $u\in D(H_{\Gamma,\max})$ has boundary values
$\displaystyle u(\gamma \pm 0) (:=\lim_{x\to\gamma \pm 0} u(x))$ and $u'(\gamma \pm 0)$ for any $\gamma\in \Gamma$.
When $d=2, 3$,
it is known that any element $u\in D(H_{\Gamma,\max})$
has asymptotics
\begin{align}
\label{intro01}
\begin{array}{cc}
u(x) = u_{\gamma,0}\log|x-\gamma| + u_{\gamma,1} + o(1)\quad 
\mbox{as }x\to \gamma\quad
& 
(d=2),\vspace{2mm}
\\
 u(x) = u_{\gamma,0}|x-\gamma|^{-1} + u_{\gamma,1} + o(1)\quad
\mbox{as }x\to \gamma\quad
&
(d=3)
\end{array}
\end{align}
for every $\gamma\in \Gamma$,
where $u_{\gamma,0}$ and $u_{\gamma,1}$ are constants
(see \cite{Al-Ge-Ho-Ho,Br-Ge-Pa} or Proposition \ref{proposition_bvalue} below).

Let $\alpha = (\alpha_\gamma)_{\gamma \in \Gamma}$ 
be a sequence of real numbers.
We define a closed linear operator $H_{\Gamma,\alpha}$ in $L^2(\mathbb{R}^d)$
by
\begin{align*}
 &H_{\Gamma,\alpha} u = -\Delta u,\\
&D(H_{\Gamma,\alpha})=
\{u\in D(H_{\Gamma,\max}) ;\ u \mbox { satisfies } (BC)_\gamma
\mbox{ for every }\gamma\in \Gamma\}.
\end{align*}
The boundary condition $(BC)_\gamma$ at the point $\gamma\in \Gamma$
is defined as follows.
\begin{align}
\label{intro02}
\begin{array}{cc}
\begin{cases}
 u(\gamma+0)= u(\gamma-0)=u(\gamma),\\ 
 u'(\gamma+0)-u'(\gamma-0)= \alpha_\gamma u(\gamma)
\end{cases}
&(d=1),\vspace{2mm}\\
 2\pi \alpha_\gamma u_{\gamma,0}+u_{\gamma,1}=0\quad &(d=2),\vspace{2mm}\\
 -4\pi \alpha_\gamma u_{\gamma,0}+u_{\gamma,1}=0\quad &(d=3),
 \\
\end{array}
\end{align}
where
$u_{\gamma,0}$ and $u_{\gamma,1}$ are the constants
in (\ref{intro01}).
The constants $2\pi$ and $-4\pi$ before coupling constants are
chosen so that the results in \cite{Al-Ge-Ho-Ho}
can be used without modification,
though our $H_{\Gamma,\alpha}$ is denoted by $-\Delta_{\alpha.Y}$ in 
\cite{Al-Ge-Ho-Ho}.
When $d=1$, 
the case $\alpha_\gamma=0$ for all $\gamma$ corresponds to the free Laplacian $H_0$,
and
the formal expression $H_{\Gamma,\alpha}=-\Delta + \sum_{\gamma\in \Gamma}\alpha_\gamma \delta_\gamma$ is justified in the sense of quadratic form,
where $\delta_\gamma$ is the Dirac delta function supported on the point $\gamma$ (see (\ref{intro05})).
However, when $d=2,3$, 
the case $\alpha_\gamma=\infty$ for all $\gamma$ 
corresponds to $H_0$,
and the coupling constant $\alpha_\gamma$ is
not the coefficient before the delta function,
but the parameter appearing in the second term 
of the expansion of the renormalization factor $\lambda(\epsilon)$
in (\ref{intro00}) (see \cite{Al-Ge-Ho-Ho}
for the detail).

It is well-known that $H_{\Gamma,\alpha}$ is self-adjoint
when $\#\Gamma<\infty$.
When $\#\Gamma=\infty$, the self-adjointness of $H_{\Gamma,\alpha}$
is proved under
the \textit{uniform discreteness condition}
\begin{equation}
 \label{intro03}
d_*:=\inf_{\gamma,\gamma'\in \Gamma,\, \gamma\not=\gamma'}|\gamma-\gamma'|>0
\end{equation}
in the book \cite{Al-Ge-Ho-Ho}
and many other references (e.g.\ \cite{Gr-Ho-Me, Ch-St,  Ge-Ma-Ch}).
There are only a few results in the case $d_*=0$.
Minami \cite{Minami} studies the self-adjointness and the spectrum
of the random Schr\"odinger operator 
$H_\omega=\displaystyle -\frac{d^2}{dt^2}+Q_t'(\omega)$ on $\mathbb{R}$,
where $\{Q_t(\omega)\}_{t\in \mathbb{R}}$ is 
a temporally homogeneous L\'evy process.
If we take 
\begin{align*}
Q_t(\omega)=\int_0^t \sum_{\gamma \in \Gamma_\omega}\alpha_{\omega,\gamma}\delta(s-\gamma)ds
\end{align*}
for the Poisson configuration 
(the support of the Poisson point process;
see Definition \ref{definition_poisson} below) 
$\Gamma_\omega$ on $\mathbb{R}$
and i.i.d.\ (independently, identically distributed) random variables
$\alpha_\omega=(\alpha_{\omega,\gamma})_{\gamma\in \Gamma_\omega}$,
we conclude that
$H_{\Gamma_\omega,\alpha_\omega}$ is self-adjoint almost surely.
Kostenko--Malamud \cite{Ko-Ma}
give the following remarkable result.
\begin{theorem}[Kostenko--Malamud \cite{Ko-Ma}]
\label{theorem_KM}
Let $d=1$.
Let $\Gamma=\{\gamma_n\}_{n\in \mathbb{Z}}$ 
be a sequence of strictly increasing real numbers with 
$\displaystyle\lim_{n\to \pm \infty}\gamma_n = \pm \infty$.
Assume 
\begin{equation}
\label{intro04}
 \sum_{n=-\infty}^{-1} d_n^2=\sum_{n=0}^\infty d_n^2=\infty,\quad
d_n=\gamma_{n+1}-\gamma_n.
\end{equation}
Then, $H_{\Gamma,\alpha}$ is self-adjoint for every 
$\alpha=(\alpha_\gamma)_{\gamma\in \Gamma}$.
\end{theorem}
\noindent
Actually Kostenko--Malamud \cite{Ko-Ma} state the result in the half-line case, 
but the result can be easily extended in the whole line case,
as stated above.
In the proof, Kostenko--Malamud 
construct an appropriate boundary triplet for
${H_{\Gamma, \min}}^*$.
Moreover, Christ--Stolz \cite{Ch-St} give a counter example
of $(\Gamma,\alpha)$ so that $d=1$, $d_*=0$ and 
$H_{\Gamma,\alpha}$ is not self-adjoint.
However, 
the proof of Minami \cite{Minami} uses that 
the deficiency indices are not more than two
for one-dimensional symmetric differential operator,
and the proof of Kostenko--Malamud \cite{Ko-Ma} uses the decomposition
$L^2(\mathbb{R})=\oplus_{n=-\infty}^\infty L^2((\gamma_{n},\gamma_{n+1}))$.
Both methods depend on the one-dimensionality of the space,
and cannot directly be applied in two or three dimensional case.

In the present paper, we give a sufficient condition
for the self-adjointness of $H_{\Gamma,\alpha}$,
which is available even in the case $d_*=0$
and $d=2,3$.
In the sequel, we denote $R$-neighborhood of a set $S$ by 
$(S)_R$, that is,
\begin{equation*}
 (S)_R:= \{ x \in \mathbb{R}^d;\ \dist(x, S)<R\},
\end{equation*}
where the distance $\dist(S,T)$ between two sets $S$ and $T$ is defined by
\begin{equation*}
 \dist(S,T):=\inf_{x\in S,\ y\in T}|x-y|.
\end{equation*}
\begin{assumption}
 \label{assumption_no_percolation}
There exists $R>0$ such that
every connected component of $(\Gamma)_R$ is a bounded set.
\end{assumption}
The set $(\Gamma)_R$ is the union of $B_R(\gamma)$,
an open disk of radius $R$ centered at $\gamma\in \Gamma$
(see Figure \ref{figure_gammaR}).
Assumption \ref{assumption_no_percolation} is
a generalization of the uniform discreteness condition
(\ref{intro03}).
Actually,
if we call the set of points of $\Gamma$ in each connected component of
$(\Gamma)_R$ a \textit{cluster},
then the assumption says `\textit{the clusters of $\Gamma$ are uniformly discrete}'.
\begin{figure}[htbp]
 \begin{center}
  \includegraphics[width=6cm]{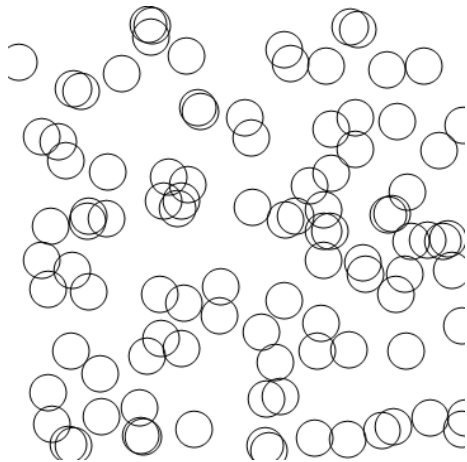}
   \caption{The set $(\Gamma)_R$ in $(-5,5)^2$ when $d=2$
and $R=0.4$.
$\Gamma$ is a sample configuration of the Poisson configuration with intensity $1$.
}
\label{figure_gammaR}
 \end{center}
\end{figure}

Our first main result is stated as follows.
\begin{theorem}
\label{theorem_main}
Let $d=1,2,3$. Suppose Assumption \ref{assumption_no_percolation} holds.
Then, $H_{\Gamma,\alpha}$ is self-adjoint for any $\alpha=(\alpha_\gamma)_{\gamma\in \Gamma}$.
\end{theorem}
\noindent
In the case $d=1$,
Theorem \ref{theorem_main} is a special case of Theorem \ref{theorem_KM},
since Assumption \ref{assumption_no_percolation} implies
there are infinitely many positive $n$ and negative $n$ such that $d_n\geq 2R$,
so the assumption (\ref{intro04}) holds.
In the case $d=2,3$, Theorem \ref{theorem_main} is new.

Theorem \ref{theorem_main} is especially useful
in the study of Schr\"odinger operators with 
random point interactions.
There are a lot of studies about the Schr\"odinger operators
with random point interactions 
(\cite{Al-Ho-Ki-Ma, De-Si-So, Minami, Bo-Gr, %
Do-Ma-Pu, Hi-Ki-Kr, Dr-Ki-Sc}),
but in most of these results $\Gamma$ is assumed to be
$\mathbb{Z}^d$ or its random subset,
except Minami's paper \cite{Minami}.
Using Theorem \ref{theorem_main},
we can study more general random point interactions
so that $d_*$ can be $0$.
In the present paper,
we prove the self-adjointness of $H_{\Gamma,\alpha}$ 
for the following two models.
First one is the \textit{random displacement model},
given as follows.
Notice that $d_*$ can be $0$ for this model.
\begin{corollary}
\label{corollary_RD}
Let $d=1,2,3$.
Let $\{\delta_n(\omega)\}_{n \in \mathbb{Z}^d}$
be a sequence of i.i.d.\ 
$\mathbb{R}^d$-valued random variables 
defined on some probability space $\Omega$
such that $|\delta_n(\omega)|< C$
for some positive constant $C$ independent of $n$ and $\omega\in \Omega$.
Put 
\begin{equation*}
 \Gamma_\omega = \{n + \delta_n(\omega)\}_{n \in \mathbb{Z}^d}.
\end{equation*}
Then, $H_{\Gamma_\omega,\alpha}$ is self-adjoint for any $\alpha=(\alpha_\gamma)_{\gamma\in \Gamma}$.
\end{corollary}
The proof of 
Corollary \ref{corollary_RD} is an application of Theorem \ref{theorem_main}
via some auxiliary result (Corollary \ref{corollary_bounded}).
Another one is the \textit{Poisson model},
given as follows.
\begin{corollary}
\label{corollary_poisson}
Let $d=1,2,3$.
Let $\Gamma_\omega$ be the Poisson configuration
on $\mathbb{R}^d$
with intensity measure $\lambda dx$ for some positive constant $\lambda$.
Then, $H_{\Gamma_\omega,\alpha}$ is self-adjoint for any $\alpha=(\alpha_\gamma)_{\gamma\in \Gamma_\omega}$, almost surely.
\end{corollary}
Corollary \ref{corollary_poisson} is proved by combining 
Theorem \ref{theorem_main} with 
the theory of continuum percolation (Theorem \ref{theorem_CP}).
These results are new when $d=2,3$,
and are not new when $d=1$, as stated before.

The proof of Theorem \ref{theorem_main} also enables us
to  determine the spectrum of $H_{\Gamma,\alpha}$
for \textit{the point interactions of Poisson-Anderson type},
defined as follows.
\begin{assumption}
 \label{assumption_PA}
\begin{enumerate}
 \item $\Gamma_\omega$ is the Poisson configuration
with intensity measure $\lambda dx$ for some $\lambda>0$.

 \item The coupling constants $\alpha_\omega=(\alpha_{\omega,\gamma})_{\gamma\in \Gamma_\omega}$
are real-valued i.i.d.\ random variables
with common distribution measure $\nu$ on $\mathbb{R}$.
Moreover, $(\alpha_{\omega,\gamma})_{\gamma\in \Gamma_\omega}$ 
are independent of $\Gamma_\omega$.
\end{enumerate}
\end{assumption}
\begin{theorem}
\label{theorem_main2}
Let $d=1,2,3$.
Let $\Gamma_\omega$ and $\alpha_\omega$ satisfy
Assumption \ref{assumption_PA},
and put $H_\omega=H_{\Gamma_\omega,\alpha_\omega}$.
Then, the spectrum $\sigma(H_\omega)$
of $H_\omega$ is given as follows.
\begin{enumerate}
 \item When $d=1$, we have
\begin{align*}
 \sigma(H_\omega)=
\begin{cases}
 [0,\infty) & (\supp \nu \subset [0,\infty)),\\
 \mathbb{R} & (\supp\nu \cap (-\infty,0)\not=\emptyset),
\end{cases}
\end{align*}
almost surely.

 \item When $d=2,3$, we have $ \sigma(H_\omega)=\mathbb{R}$ almost surely.
\end{enumerate}
\end{theorem}
\noindent
Notice that there is \textit{no assumption} on $\supp \nu$ when $d=2,3$.
Theorem \ref{theorem_main2} can be interpreted as a generalization
of the corresponding result 
for the Schr\"odinger operator 
$-\Delta+V_\omega$
with random \textit{scalar} potential of Poisson-Anderson type
\begin{align*}
V_\omega(x)=\sum_{\gamma\in \Gamma_\omega} \alpha_{\omega,\gamma} V_0(x-\gamma),
\end{align*}
where $\Gamma_\omega$ and $\alpha_\omega$ satisfy
Assumption \ref{assumption_PA},
and 
$V_0$ is a real-valued scalar function having some regularity 
and decaying property.
The spectrum $\sigma(-\Delta+V_\omega)$ is determined 
in \cite{Pa-Fi,An-Iw-Ka-Na,Ka-Mi},
and the result says \textit{`the spectrum equals $[0,\infty)$
if $V_\omega$ is non-negative,
and it equals $\mathbb{R}$ if $V_\omega$ has negative part'}.
When $d=1$,
the point interaction at $\gamma$ has the same sign
as the sign of the coupling constant $\alpha_\gamma$ in the sense of quadratic form,
that is,
\begin{align}
\label{intro05}
 (u,H_{\Gamma,\alpha}u)=\|\nabla u\|^2 + \sum_{\gamma\in \Gamma}\alpha_\gamma |u(\gamma)|^2
\end{align}
for $u\in D(H_{\Gamma,\alpha})$ with bounded support.
When $d=2,3$, the sign of point interaction at $\gamma$ is 
in some sense \textit{negative}
for any $\alpha_\gamma\in \mathbb{R}$.
Actually, in the approximation (\ref{intro00}),
the limiting operator $H$ is not the free operator $H_0$
only if the zero-energy resonance of $-\Delta + V$ exists,
and the existence of the zero-energy resonance requires $V$ has negative part
(see \cite{Al-Ge-Ho-Ho}).
There is also qualitative difference between the proof of Theorem \ref{theorem_main2}
in the case $d=1$ and that in the case $d=2,3$.
The spectrum $(-\infty,0)$ is created by
the accumulation of many points in one place when $d=1$, 
while it is created by 
the meeting of two points when $d=2,3$
(see section 3.3).
The latter fact reminds us \textit{Thomas collapse},
which says the mass defect of the tritium ${^3}{\rm H}$
becomes arbitrarily large as 
the distances between a proton and two neutrons become small enough
(see \cite{Th}).

Let us give a brief comment on the magnetic case.
The Schr\"odinger operator with a constant magnetic field
plus infinitely many point interactions is studied in \cite{Ge, Do-Ma-Pu},
and the self-adjointness is proved under the uniform discreteness 
condition (\ref{intro03}).
Theorem \ref{theorem_main} can be generalized under the existence
of a constant magnetic field, 
by using the magnetic translation operator.
We will discuss this case elsewhere in the near future.

The present paper is organized as follows.
In section 2, 
we review some fundamental formulas about self-adjoint extensions
of $H_{\Gamma, \min}$, and prove Theorem \ref{theorem_main}.
The crucial fact is
`bounded support functions are dense in 
$D(H_{\Gamma,\alpha})$ under Assumption \ref{assumption_no_percolation}'
(Proposition \ref{proposition_dense}).
In section 3, 
we prove the self-adjointness of Schr\"odinger operators
with various random point interactions.
We also determine the spectrum of $H_\omega=H_{\Gamma_\omega,\alpha_\omega}$
for Poisson--Anderson type point interactions,
using the method of \textit{admissible potentials}
(Proposition \ref{proposition_admissible}; 
see also \cite{Ki-Ma,Pa-Fi,An-Iw-Ka-Na,Ka-Mi}).
In the proof, we again need Proposition \ref{proposition_dense},
and also need to take care of the dependence of the operator domain $D(H_{\Gamma,\alpha})$
with respect to $\Gamma$ and $\alpha$.
Once we establish the method of admissible potentials,
the proof of Theorem \ref{theorem_main2} is reduced to
the calculation of $\sigma(H_{\Gamma,\alpha})$ for admissible $(\Gamma,\alpha)$.

Let us explain the notation in the manuscript.
The notation $A:=B$ means $A$ is defined as $B$.
The set $B_r(x)$ is 
the open ball of radius $r$ centered at $x\in \mathbb{R}^d$,
that is, $B_r(x):=\{y\in \mathbb{R}^d;|y-x|<r\}$.
The space $D(H)$ is the operator domain of a linear operator $H$
equipped with the graph norm $\|u\|_{D(H)}^2=\|u\|^2 + \|H u\|^2$.
For an open set $U$, 
$C_0^\infty(U)$ is the set of compactly supported $C^\infty$ functions
on $U$.
The space $L^2(U)$ is the space of square integrable functions
on $U$, and the inner product and the norm on $L^2(U)$
are defined as
\begin{align*}
(u,v)_{L^2(U)}=\int_U \overline{u}{v}dx,\quad
 \|u\|_{L^2(U)}=(u,u)_{L^2(U)}^{1/2}=\left(\int_U |u|^2dx\right)^{1/2}.
\end{align*}
When $U=\mathbb{R}^d$, we often abbreviate the suffix $L^2(U)$.
The space $H^2(U)$ is the Sobolev space of order $2$ on $U$,
and the norm is defined by
\begin{align*}
 \|u\|_{H^2(U)}^2=\sum_{0\leq |\alpha|\leq 2}
\left\|\frac{\partial^{|\alpha|} u}{\partial x_1^{\alpha_1}\cdots\partial x_d^{\alpha_d}}\right\|_{L^2(U)}^2,
\end{align*}
where $\alpha=(\alpha_1,\ldots,\alpha_d)\in (\mathbb{Z}_{\geq 0})^d$
is the multi-index and $|\alpha|=\alpha_1+\cdots+\alpha_d$,
and the derivatives are defined as 
elements of $\mathcal{D}'(U)$,
the Schwartz distributions on $U$.
The space $L^2_{\rm loc}(U)$ is the set of the functions $u$
such that $\chi u\in L^2(U)$ for any $\chi\in C_0^\infty(U)$.
The space $H^2_{\rm loc}(U)$ is defined similarly.

\section{Self-adjointness}
\subsection{Structure of $D(H_{\Gamma, \max})$}
First we review fundamental properties of the operator domain
$D(H_{\Gamma, \max})$ of the maximal operator $H_{\Gamma, \max}$.
Most of the results are already obtained
under more general assumption (see e.g.\ \cite{Al-Ge-Ho-Ho,Br-Ge-Pa}),
but we prove them here again for the completeness of the present manuscript.
\begin{proposition}
\label{proposition_hmax}
Let $d=1,2,3$.
 \begin{enumerate}
  \item We have
\begin{align*}
 D(H_{\Gamma, \max})
&=\{u\in L^2(\mathbb{R}^d)
\,;\, \Delta u \in L^2(\mathbb{R}^d)\}\\
&=\{u\in L^2(\mathbb{R}^d)\cap H^2_{\rm loc}(\mathbb{R}^d\setminus \Gamma)
\,;\, \Delta u \in L^2(\mathbb{R}^d)\},
\end{align*}
where $\Delta u$ is defined as an element of
$\mathcal{D}'(\mathbb{R}^d\setminus\Gamma)$.

  \item Let $\chi\in C_0^\infty(\mathbb{R}^d)$ such that
$\Gamma \cap \supp \nabla\chi =\emptyset$.
Then, for any $u\in D(H_{\Gamma, \max})$, we have
$\chi u\in D(H_{\Gamma, \max})$.
 \end{enumerate}
\end{proposition}
\begin{proof}
 (i) By definition, the statement
$u\in D(H_{\Gamma,\max})=D({H_{\Gamma,\min}}^*)$
is equivalent to  `$u\in L^2(\mathbb{R}^d)$ and there exists
$v\in L^2(\mathbb{R}^d)$ such that
\begin{align*}
 (u, -\Delta\phi)=(v, \phi)
\end{align*}
for any $\phi\in C_0^\infty(\mathbb{R}^d\setminus \Gamma)$'.
The latter statement is equivalent to $v=-\Delta u\in L^2(\mathbb{R}^d)$,
where $\Delta u$ is defined as an element of 
$\mathcal{D}'(\mathbb{R}^d\setminus\Gamma)$.
Moreover, by the elliptic inner regularity theorem 
(Corollary \ref{corollary_agmon}), we have $u\in H^2_{\rm loc}(\mathbb{R}^d\setminus \Gamma)$.

(ii) 
Let $\chi$ satisfy the assumption, and $u\in D(H_{\Gamma,\max})$.
By the chain rule,
we have
\begin{align}
\label{2-1-01}
 \Delta(\chi u)=(\Delta\chi)u+2 \nabla\chi\cdot \nabla u +\chi\Delta u.
\end{align}
Since $u,\Delta u\in L^2(\mathbb{R}^d)$, the first term of (\ref{2-1-01})
and the third
belong to $L^2(\mathbb{R}^d)$.
Moreover, since $\supp \nabla \chi$ is a compact subset of 
$\mathbb{R}^d\setminus\Gamma$ and 
$u\in H^2_{\rm loc}(\mathbb{R}^d\setminus\Gamma)$,
the second term also belongs to $L^2(\mathbb{R}^d)$.
Thus $\chi u, \Delta(\chi u)\in L^2(\mathbb{R}^d)$,
and the statement follows from (i).
\end{proof}
\noindent
The assumption $\Gamma \cap \supp\nabla \chi =\emptyset$ above cannot 
be removed when $d=2,3$. For example,
consider the case $d=2$, $\Gamma=O:=\{0\}$.
Take functions $u\in C^\infty(\mathbb{R}^2\setminus O)$ 
and $\chi\in C_0^\infty(\mathbb{R}^2)$ 
such that
\begin{align}
 u(x)=
\begin{cases}
 \log |x| & |x|<1,\cr
 0 & |x|>2,
\end{cases}
\quad
\chi(x)= 
\begin{cases}
x_1 & |x|<1,\cr
 0 & |x|>2.
\end{cases}
\end{align}
Since $\Delta \log|x|=0$ for $x\not=0$,
we see that $u,\Delta u\in L^2(\mathbb{R}^2)$, so $u\in D(H_{O,\max})$.
However, the chain rule (\ref{2-1-01}) implies
\begin{align*}
 \Delta(\chi u) = \frac{2x_1}{|x|^2}
\end{align*}
for $|x|<1$, so $\Delta(\chi u)\not\in L^2(\mathbb{R}^2)$.
This fact is crucial in our proof of self-adjointness criterion
(Theorem \ref{theorem_main}).

Next we define the (generalized) boundary values at $\gamma\in \Gamma$
of $u\in D(H_{\Gamma,\max})$.
In the case $d=2,3$,
similar argument is found in \cite{Al-Ge-Ho-Ho, Br-Ge-Pa}.
\begin{proposition}
\label{proposition_bvalue}
\begin{enumerate}
 \item Let $d=1$,
$u\in D(H_{\Gamma, \max})$ and $\gamma \in \Gamma$.
Then,
one-side limits $u(\gamma\pm 0)$ and $u'(\gamma\pm 0)$ exist.

 \item 
Let $d=2,3$, $u\in D(H_{\Gamma, \max})$ and $\gamma\in \Gamma$.
Let $\epsilon$ be a small positive constant
so that $\overline{B_\epsilon(\gamma)}\cap \Gamma=\{\gamma\}$.
Then, 
there exist unique constants $u_{\gamma,0}$ and $u_{\gamma,1}$,
and $\widetilde{u}\in H^2(B_\epsilon(\gamma))$ with $\widetilde{u}(\gamma)=0$,
such that for $x\in B_{\epsilon}(\gamma)$
\begin{equation}
\label{2-1-02}
\begin{array}{cc}
u(x) = u_{\gamma,0}\log|x-\gamma|+u_{\gamma,1}+ \widetilde{u}(x)
& (d=2),\\
u(x) = u_{\gamma,0}|x-\gamma|^{-1}+u_{\gamma,1}+ \widetilde{u}(x)
& (d=3).
\end{array}
\end{equation}

\end{enumerate}
\end{proposition}
\begin{proof}
(i) This is a consequence of the Sobolev embedding theorem,
since the restriction of $u\in D(H_{\Gamma,\max})$ 
on $I$ belongs to
$H^2(I)$ 
for any connected component $I$ of $\mathbb{R}\setminus\Gamma$.

(ii)
We consider the case $d=2$.
By a cut-off argument ((ii) of Proposition \ref{proposition_hmax}),
we can reduce the proof to the case $\Gamma$ equals one point set.
 Without loss of generality, we assume $\Gamma=O$.
Then, by von Neumann's theory of self-adjoint extensions
(see e.g.\ \cite[Section X.1]{Re-Si}),
we have
\begin{align}
\label{2-1-03}
 D(H_{O,\max}) = 
\overline{D(H_{O,\min})}
\oplus \mathcal{K}_{O,-}
\oplus\mathcal{K}_{O,+},
\end{align}
where $\overline{D(H_{O,\min})}$ is the closure of $D(H_{O,\min})$
with respect to the graph norm (or $H^2$-norm),
and $\mathcal{K}_{O,\pm} = \Ker(H_{O,\max}\mp i)$ are deficiency subspaces.
It is known that
$\mathcal{K}_{O,\pm}$ are one dimensional spaces spanned by
\begin{align*}
\varphi_\pm(x) = H^{(1)}_0(\sqrt{\pm i}r),
\end{align*}
where $H^{(1)}_0$ is the 0-th order Hankel function of the first kind,
$r=|x|$, and the branches of $\sqrt{\pm i}$ are taken as $\Im \sqrt{\pm i}>0$
(see \cite{{Al-Ge-Ho-Ho}}).
Thus we have inclusion 
\begin{align}
\label{2-1-04}
 \overline{D(H_{O,\min})}
\subset \{u\in H^2(\mathbb{R}^2)\,;\,u(0)=0\}
\subset H^2(\mathbb{R}^2)
\subset
D(H_{O,\max}).
\end{align}
The first inclusion is due to the Sobolev embedding theorem.
The second inclusion is clearly strict,
and the third one is also strict
since $D(H_{O,\max})$ contains elements singular at $0$,
by (\ref{2-1-03})
(see (\ref{2-1-05}) below).
The decomposition (\ref{2-1-03}) also tells us 
$\dim \left(D(H_{O,\max})/ \overline{D(H_{O,\min})}\right)=2$,
so the first inclusion in (\ref{2-1-04}) must be equality, that is,
\begin{align*}
  \overline{D(H_{O,\min})} = \{u\in H^2(\mathbb{R}^2)\,;\,u(0)=0\}.
\end{align*}

By the series expansion of the Hankel function,
we have
\begin{align}
\label{2-1-05}
 \varphi_\pm(x) 
= 
1 + \frac{2i}{\pi}\left(\gamma_E+\log \frac{\sqrt{\pm i}r}{2}\right)
+O(r^2 \log r)\quad (r\to 0),
\end{align}
where $\gamma_E$ is the Euler constant.
It is easy to see the remainder term is in $H^2(B_\epsilon(0))$ 
and vanishes at $0$.
Thus by the decomposition (\ref{2-1-03}),
every $u\in D(H_{O,\max})$ can be uniquely written as (\ref{2-1-02}).

In the case $d=3$,
a basis of the deficiency subspace $\mathcal{K}_{O,\pm}$ is
\begin{align*}
 \varphi_\pm(x) = \frac{e^{i\sqrt{\pm i}r}}{r}
=
\frac{1}{r}+i\sqrt{\pm i}+O(r)
\quad (r\to 0)
\end{align*}
(see \cite{{Al-Ge-Ho-Ho}}).
Using this expression,
we can prove the statement for $d=3$ similarly.
\end{proof}

Next we introduce the generalized Green formula.
\begin{proposition}
 \label{proposition_green}
Let $d=1,2,3$.
Let $u,v \in D(H_{\Gamma, \max})$,
and assume $\supp u$ or $\supp v$ is bounded.
Then we have
\begin{align}
\label{2-1-06}
& (H_{\Gamma, \max}u,v)- (u,H_{\Gamma, \max}v)\notag\\
&=
\begin{cases}
\displaystyle
\sum_{\gamma \in \Gamma} 
\bigl(
-\overline{u'(\gamma-0)}v(\gamma-0)
+\overline{u(\gamma-0)}v'(\gamma-0)\\
\quad\quad+\overline{u'(\gamma+0)}v(\gamma+0)
-\overline{u(\gamma+0)}v'(\gamma+0)
\bigr)
&
(d=1),\\
\displaystyle
\sum_{\gamma \in \Gamma} 2\pi 
(
\overline{u_{\gamma,0}}v_{\gamma,1}
-
\overline{u_{\gamma,1}}v_{\gamma,0}
) 
& 
(d=2),\\
\displaystyle
\sum_{\gamma \in \Gamma} (-4\pi)
(
\overline{u_{\gamma,0}}v_{\gamma,1}
-
\overline{u_{\gamma,1}}v_{\gamma,0}
) 
& 
(d=3).
\end{cases}
\end{align}
\end{proposition}
\begin{proof}
The proof in the case $d=1$ is easy.
Consider the case $d=2$.
By a cut-off argument,
we can assume both $\supp u$ and $\supp v$ are bounded.
We can also assume
$\supp u \cup \supp v\subset B_R(0)$,
and $\Gamma \cap \partial B_R(0)=\emptyset$.
Then, we can decompose $u$ and $v$ as
\begin{align*}
 u=
\sum_{\gamma\in\Gamma\cap B_R(0)}
(u_{\gamma,0}\phi_\gamma
+
u_{\gamma,1}\psi_\gamma
)
+
\widetilde{u},\quad
 v=
\sum_{\gamma\in\Gamma\cap B_R(0)}
(v_{\gamma,0}\phi_\gamma
+
v_{\gamma,1}\psi_\gamma
)
+
\widetilde{v},
\end{align*}
where $\widetilde{u},\widetilde{v}\in \overline{D(H_{\Gamma,\min})}$,
and $\phi_\gamma,\psi_\gamma\in D(H_{\Gamma,\max})$ are real-valued functions 
such that
\begin{align*}
 \phi_\gamma(x)=\log|x-\gamma|,\quad
 \psi_\gamma(x)=1\quad \mbox{near }x=\gamma,
\end{align*}
and $\supp\phi_\gamma\cup \supp\psi_\gamma$ 
is contained in some small neighborhood of $\gamma$
so that 
$\{\supp\phi_\gamma\cup \supp\psi_\gamma\}_{\gamma \in \Gamma\cap B_R(0)}$ 
are disjoint sets in $B_R(0)$.

We use the notation
\begin{align*}
 [\phi,\psi]=(H_{\Gamma, \max}\phi,\psi)- (\phi,H_{\Gamma, \max}\psi).
\end{align*}
Clearly $[\phi,\psi]=-\overline{[\psi,\phi]}$,
so $[\phi,\phi]$=0 for real-valued $\phi\in D(H_{\Gamma,\max})$.
Moreover, $[\phi,\psi]=0$ 
if  $\phi\in D(H_{\Gamma,\max})$ and $\psi\in \overline{D(H_{\Gamma,\min})}$.
Thus we have
\begin{align*}
[u,v]=
\sum_{\gamma\in\Gamma\cap B_R(0)}
\left(
\overline{u_{\gamma,0}}v_{\gamma,1}
-
\overline{u_{\gamma,1}}v_{\gamma,0}
\right)
[\phi_\gamma,\psi_\gamma].
\end{align*}
Let us calculate $[\phi_\gamma,\psi_\gamma]$.
By translating the coordinate,
we assume $\gamma=0$,
and write $\phi_\gamma=\phi$, $\psi_\gamma=\psi$.
Then, since $\phi=\log r$ and $\psi=1$ near $x=0$,
\begin{align*}
[\phi,\psi]
=&
\lim_{\epsilon\downarrow 0}
\int_{B_\epsilon(0)^c}
\left((-\overline{\Delta \phi})\psi + \overline{\phi}(\Delta \psi)\right)dx\\
=&
\lim_{\epsilon\downarrow 0}
\int_{\partial B_\epsilon(0)}
\left((-\overline{\nabla \phi\cdot n})\psi + \overline{\phi}(\nabla \psi\cdot n)\right)ds\\
=&
\lim_{r\downarrow 0}
\int_0^{2\pi}
\left(\frac{\partial \phi}{\partial r}\cdot \psi - \phi\cdot \frac{\partial \psi}{\partial r}\right)rd\theta\\
=&2\pi,
\end{align*}
where $n$ is the unit inner normal vector on $\partial B_\epsilon(0)$,
$ds$ is the line element,
and $(r,\theta)$ is the polar coordinate.
Thus the assertion for $d=2$ holds.
The proof for the case $d=3$ is similar, 
but we take the function $\phi_\gamma$ as
\begin{align*}
 \phi_\gamma(x)=|x-\gamma|^{-1}\quad \mbox{near }x=\gamma.
\end{align*}
\end{proof}
If the uniform discreteness condition (\ref{intro03}) holds,
the results in this subsection can be formulated in terms of 
the boundary triplet for $H_{\Gamma,\max}$,
as is done in \cite{Br-Ge-Pa}.
When $d=1$ and $d_*=0$,
the boundary triplet for $H_{\Gamma,\max}$ 
is constructed in \cite{Ko-Ma}.
The construction in the case $d=2,3$ and $d_*=0$ seems to be unknown so far.

\subsection{Proof of Theorem \ref{theorem_main}}
Let $\Gamma$ be a locally finite discrete set in $\mathbb{R}^d$,
and  $\alpha=(\alpha_\gamma)_{\gamma\in \Gamma}$ be
a sequence of real numbers.
In this subsection we write $H=H_{\Gamma,\alpha}$,
that is,
\begin{align*}
 &Hu=-\Delta u,\\
&D(H)=\{u\in D(H_{\Gamma,\max})\,;\, u \mbox{ satisfies }(BC)_\gamma 
\mbox{ for every }\gamma\in \Gamma
\},
\end{align*}
where $(BC)_\gamma$ is defined in (\ref{intro02}).
We introduce an auxiliary operator $H_b$ by
\begin{align*}
 H_b u = -\Delta u,
\quad
D(H_b)=\{ u\in D(H)\, ;\, \mbox{$\supp u$ is bounded}\}.
\end{align*}
By the generalized Green formula (Proposition \ref{proposition_green}),
we have the following.
\begin{proposition}
\label{proposition_symmetric}
Let $d=1,2,3$.
For any $\Gamma$ and $\alpha$,
the operator $H_b$ is a densely defined symmetric operator, and ${H_b}^*=H$.
\end{proposition}
\begin{proof}
We consider the case $d=2$,
since the case $d=1,3$ can be treated similarly.
For $u,v\in D(H_b)$,
the generalized Green formula (\ref{2-1-06}) and $(BC)_\gamma$ imply
\begin{align}
 [u,v]:=&
(H u,v)-(u,Hv)
=
\sum_{\gamma \in \Gamma} 2\pi 
(
\overline{u_{\gamma,0}}v_{\gamma,1}
-
\overline{u_{\gamma,1}}v_{\gamma,0}
) \notag\\
=&
\sum_{\gamma \in \Gamma} 
(2\pi)^2\alpha_\gamma
(
-
\overline{u_{\gamma,0}}v_{\gamma,0}
+
\overline{u_{\gamma,0}}v_{\gamma,0}
) 
=
0.
\label{2-2-01}
\end{align}
Thus $H_b$ is a symmetric operator.

The equality (\ref{2-2-01}) also holds for any $u\in D(H_b)$
and $v\in D(H)$, so $D(H)\subset D({H_b}^*)$.
Conversely, let $v\in D({H_b}^*)$.
By definition,
$[u,v]=0$ holds for any $u\in D(H_b)$.
For $\gamma\in \Gamma$, 
take $u\in D(H_b)$ such that 
$u_{\gamma,0}=1/(2\pi)$, $u_{\gamma,1}=-\alpha_\gamma$,
and $u_{\gamma',0}=u_{\gamma',1}=0$ for $\gamma'\not=\gamma$.
Since $D({H_b}^*)\subset D({H_{\Gamma,\min}}^*)=D(H_{\Gamma,\max})$,
we have by the generalized Green formula (\ref{2-1-06})
\begin{align*}
[u,v]= v_{\gamma,1}+2\pi \alpha_\gamma v_{\gamma,0}=0.
\end{align*}
Thus $v$ satisfies $(BC)_\gamma$ for every $\gamma\in \Gamma$,
and we conclude $v\in D(H)$.
This means $H={H_b}^*$.
\end{proof}
Now Theorem \ref{theorem_main} is a corollary of the following proposition.
\begin{proposition}
 \label{proposition_dense}
Suppose Assumption \ref{assumption_no_percolation} holds.
Then, $\overline{H_b}=H$. In other words,
$D(H_b)$ is an operator core for the operator $H$.
\end{proposition}
\begin{proof}
 Let $R$ be the constant in Assumption \ref{assumption_no_percolation}.
For a positive integer $n$, let $S_n$ be the connected component
of $B_n(0)\cup (\Gamma)_R$ containing $B_n(0)$
(see Figure \ref{figure_Sn}).

\begin{figure}[htbp]
 \begin{center}
  \includegraphics[width=6cm]{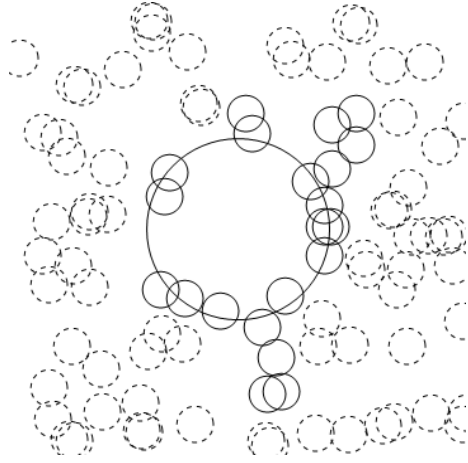}
   \caption{
The union of non-dashed disks is the set $S_n$
for $n=2$.
Here $(\Gamma)_R$ is the set in Figure \ref{figure_gammaR}.
}
\label{figure_Sn}
 \end{center}
\end{figure}

By assumption,
$S_n$ is a bounded open set in $\mathbb{R}^d$.
Let $\eta\in C_0^\infty(\mathbb{R}^d)$ 
be a rotationally symmetric function such that
$\eta\geq 0$, 
$\supp \eta\subset B_{R/3}(0)$,
and $\int_{\mathbb{R}^d}\eta dx =1$.
Put
\begin{align*}
 \chi_n(x)=\int_{Sn}\eta(x-y)dy.
\end{align*}
The function $\chi_n$ has the following properties.
\begin{enumerate}
 \item $\chi_n\in C_0^\infty(\mathbb{R}^d)$, $0\leq \chi_n(x)\leq 1$,
and
\begin{align*}
 \chi_n(x)=
\begin{cases}
1 & (x\in S_n,\ \dist(x,\partial S_n)>R/3),  \\
0 & (x\not\in S_n,\ \dist(x,\partial S_n)>R/3).
\end{cases}
\end{align*}

In particular, $\chi_n(x)\to 1$ as $n\to \infty$ for every $x\in \mathbb{R}^d$.

 \item $\supp \nabla \chi_n\subset (\partial S_n)_{R/3}$,
and $\supp \nabla \chi_n\cap \Gamma=\emptyset$.

 \item 
$\|\nabla \chi_n\|_\infty$,
$\|\Delta \chi_n\|_\infty$ are bounded uniformly with respect to $n$,
where $\|\cdot\|_\infty$ denotes the sup norm.
\end{enumerate}
Let $u\in D(H)$.
By (i), (ii) and Proposition \ref{proposition_hmax}, $\chi_n u \in D(H_b)$.
By the dominated convergence theorem and (i),
$\chi_n u \to u$ in $L^2(\mathbb{R}^d)$.
Moreover,
\begin{align}
\label{2-2-02}
 \Delta(\chi_n u)-\Delta u
= (\chi_n-1)\Delta u + 2\nabla \chi_n\cdot \nabla u + (\Delta \chi_n) u.
\end{align}
Since $u,\Delta u\in L^2$,
the first term of (\ref{2-2-02}) and the third tend to $0$ in $L^2(\mathbb{R}^d)$
by the dominated convergence theorem.
As for the second term, we 
apply the elliptic inner regularity estimate
(Corollary \ref{corollary_agmon})
for $U=(\partial S_n)_{R/2}$ and $V=(\partial S_n)_{R/3}$,
and obtain
\begin{align*}
& \|\nabla \chi_n\cdot \nabla u\|_{L^2(\mathbb{R}^d)}\\
\leq& \|\nabla \chi_n\|_\infty
\|\nabla u\|_{L^2(V)}\\
\leq& C \|\nabla \chi_n\|_\infty
\left(\|\Delta u\|_{L^2(U)}
+
\|u\|_{L^2(U)}
\right)\\
\leq& C \|\nabla \chi_n\|_\infty
\left(\|\Delta u\|_{L^2(B_{n-R/2}(0)^c)}
+
\|u\|_{L^2(B_{n-R/2}(0)^c)}
\right).
\end{align*}
Here the constant $C$ is independent of $n$, 
since $\dist (V,\partial U)\geq R/6$ and the lower bound
is independent of $n$.
The last expression tends to $0$ as $n\to \infty$,
so $\Delta(\chi_n u)-\Delta u\to 0$ in $L^2(\mathbb{R}^d)$.
Thus $\chi_n u \in D(H_b)$ converges to $u$ in $D(H)$,
and we conclude $D(H_b)$ is dense in $D(H)$.
\end{proof}
\begin{proof}[Proof of Theorem \ref{theorem_main}]
Proposition \ref{proposition_symmetric} implies 
$H={H_b}^*$, and ${H_b}^*=(\overline{H_b})^*$ always holds.
On the other hand, Proposition \ref{proposition_dense} says
$\overline{H_b}=H$, so
\begin{align*}
 H = {H_b}^* = (\overline{H_b})^* = H^*.
\end{align*}
Thus $H$ is self-adjoint.
\end{proof}

\section{Random point interactions}
Using Theorem \ref{theorem_main},
we study the Schr\"odinger operators 
with random point interactions so that $d_*$ can be $0$.

\subsection{Self-adjointness}
First we give a simple corollary of Theorem \ref{theorem_main}.
\begin{corollary}
\label{corollary_bounded}
Assume that there exists $R_0>0$ and $M>0$ such that
$\#(\Gamma \cap B_{R_0}(x))\leq M$ 
for every $x\in \mathbb{R}^d$.
Then, $H_{\Gamma,\alpha}$ is self-adjoint for any $\alpha=(\alpha_\gamma)_{\gamma\in \Gamma}$.
\end{corollary}
\begin{proof}
The assumption 
implies Assumption \ref{assumption_no_percolation}
holds with $R=R_0/(2M)$,
since the connected component of 
$(\Gamma)_R$ containing $x\in \mathbb{R}^d$
is contained in the bounded set $B_{R_0}(x)$.
\end{proof}
The assumption of Corollary \ref{corollary_bounded} is satisfied
for \textit{random displacement model} (Corollary \ref{corollary_RD}).

\begin{proof}[Proof of Corollary \ref{corollary_RD}]
Under the assumption of Corollary \ref{corollary_RD},
we have
\begin{equation*}
\#(\Gamma\cap B_1(x)) \leq 
\#(\mathbb{Z}^d\cap B_{C+1}(x))
\leq
\left|B_{C+1+\sqrt{d}/2}(0)\right|,
\end{equation*}
where $|S|$ denotes the Lebesgue measure of a measurable set $S$.
Thus the assumption of Corollary \ref{corollary_bounded}
is satisfied.
\end{proof}

Next, we consider the case $\Gamma=\Gamma_\omega$ is the Poisson configuration
(Corollary \ref{corollary_poisson}).
We review the definition of the Poisson configuration (see e.g.\ \cite{Pa-Fi,An-Iw-Ka-Na,Ka-Mi,Re}).
\begin{definition}
\label{definition_poisson}
Let $\mu_\omega$ be a random measure on $\mathbb{R}^d$ ($d\geq 1$)
dependent on $\omega\in \Omega$ for some probability space $\Omega$.
For a positive constant $\lambda$,
we say $\mu_\omega$ is the Poisson point process
with intensity measure $\lambda dx$ 
if the following conditions hold.
\begin{enumerate}
 \item For every Borel measurable set $E\in \mathbb{R}^d$ 
with the Lebesgue measure $|E|<\infty$, 
$\mu_\omega(E)$ is an integer-valued random variable on $\Omega$ and
\begin{align*}
 \mathbb{P}(\mu_\omega(E)=k)=\frac{(\lambda|E|)^k}{k!}e^{-\lambda|E|}
\end{align*}
for every non-negative integer $k$.

 \item 
For any disjoint Borel measurable sets $E_1,\ldots,E_n$ in $\mathbb{R}^d$
with finite Lebesgue measure,
the random variables $\{\mu_\omega(E_j)\}_{j=1}^n$
are independent.
\end{enumerate}
We call the support $\Gamma_\omega$
of the Poisson point process measure $\mu_\omega$
the \textit{Poisson configuration}.
\end{definition}
We introduce a basic result in the theory of continuum percolation
(see e.g.\ \cite{Me-Ro}).
\begin{theorem}[Continuum percolation]
\label{theorem_CP}
Let $\Gamma=\Gamma_\omega$
be the Poisson configuration on $\mathbb{R}^d$ ($d\geq 2$)
with intensity measure $\lambda dx$,
where $\lambda$ is a positive constant.
For $R>0$, let $\theta_R(\lambda)$ be the probability of the event
`the connected component of $(\Gamma)_R$ containing the origin
is unbounded'.
Then,
for any $R>0$,
there exists a positive constant $\lambda_c(R)$,
called the \textit{critical density},
such that 
\begin{align*}
\begin{cases}
\theta_R(\lambda)=0  & (\lambda < \lambda_c(R)),\cr
\theta_R(\lambda)>0  & (\lambda > \lambda_c(R)).
\end{cases}
\end{align*}
Moreover, the scaling property
\begin{equation}
\label{3-1-01}
 \lambda_c(R)=R^{-d}\lambda_c(1)
\end{equation}
holds for any $R>0$.
\end{theorem}
When $d=1$, it is easy to see $\theta_R(\lambda)=0$ for every $R>0$ and $\lambda>0$,
so we put $\lambda_c(R)=\infty$.

\begin{proof}[Proof of Corollary \ref{corollary_poisson}]
By the scaling property (\ref{3-1-01}),
the condition $\lambda<\lambda_c(R)$ is satisfied
if we take $R$ sufficiently small.
Then, since the Poisson point process is statistically translationally invariant
and $\mathbb{R}^d$ has a countable dense subset,
we see that every connected component of $(\Gamma_\omega)_R$ is bounded,
almost surely.
Thus Theorem \ref{theorem_main} implies the conclusion.
\end{proof}

\subsection{Admissible potentials for Poisson-Anderson type point interactions}
By
Corollary \ref{corollary_poisson},
we can define 
the Schr\"odinger operator with random point interactions 
of \textit{Poisson-Anderson type}, that is,
$(\Gamma_\omega,\alpha_\omega)$
satisfies Assumption \ref{assumption_PA}.
We write $H_\omega=H_{\Gamma_\omega,\alpha_\omega}$ for simplicity,
and study the spectrum of $H_\omega$.
For this purpose,
we use the method of \textit{admissible potentials},
which is a useful method when we determine the spectrum
of the random Schr\"odinger operators (see e.g.\ 
\cite{Ki-Ma,Pa-Fi,An-Iw-Ka-Na,Ka-Mi}).

\begin{definition}
Let $\nu$ be the single-site measure in (ii) of Assumption \ref{assumption_PA}.
 \begin{enumerate}
  \item We say a pair $(\Gamma,\alpha)$ belongs to ${\cal A}_F$
if
$\Gamma$ is a finite set in $\mathbb{R}^d$
and $\alpha=(\alpha_\gamma)_{\gamma\in \Gamma}$ 
with $\alpha_\gamma\in \supp \nu$
for every $\gamma\in \Gamma$.

  \item We say a pair $(\Gamma,\alpha)$ belongs to ${\cal A}_P$
if
$\Gamma$ is expressed as
\begin{align*}
 \Gamma=\bigcup_{k=1}^n
\left(\gamma_k+
\bigoplus_{j=1}^d \mathbb{Z}e_j
\right)
\end{align*}
for 
some $n=0,1,2,\ldots$,
some vectors
$\gamma_1,\ldots,\gamma_n\in \mathbb{R}^d$
and independent vectors $e_1,\ldots,e_d\in \mathbb{R}^d$,
and $\alpha=(\alpha_\gamma)_{\gamma\in \Gamma}$ 
is a $\supp\nu$-valued periodic sequence on $\Gamma$,
i.e., $\alpha_\gamma \in \supp\nu$ for every $\gamma\in \Gamma$
and $\alpha_{\gamma+e_j}=\alpha_\gamma$ for every $\gamma\in \Gamma$
and $j=1,\ldots,d$.
 \end{enumerate}
\end{definition}
\noindent
Notice that $(\Gamma,\alpha)$
belongs to both $\mathcal{A}_F$ and $\mathcal{A}_P$
if $\Gamma=\emptyset$.

We need a lemma about the continuous dependence of 
the operator domain
$D(H_{\Gamma,\alpha})$ with respect to $(\Gamma,\alpha)$.
\begin{lemma}
\label{lemma_approximatingEF}
Let 
$\Gamma=\{\gamma_j\}_{j=1}^n$ be an $n$-point set
and $\alpha=(\alpha_j)_{j=1}^n$ a real-valued sequence on $\Gamma$,
where we denote $\alpha_{\gamma_j}$ by $\alpha_j$.
Let $\delta=\min_{j\not=k}|\gamma_j-\gamma_k|$.
Let $\epsilon>0$, 
$E\in \mathbb{R}$,
and $U$ be a bounded open set.
Suppose that there exists
$u_\epsilon \in D(H_{\Gamma,\alpha})$ such that
$\|u_\epsilon\|=1$,
$\supp u_\epsilon\subset U$, and
$\|(H_{\Gamma,\alpha}-E)u_\epsilon\|\leq \epsilon$.
Then the following holds.

\begin{enumerate}
\item 
There exists $\epsilon'>0$ satisfying the following property;
for any $\widetilde{\Gamma}=\{\widetilde{\gamma}_j\}_{j=1}^n$ with
$|\gamma_j-\widetilde{\gamma}_j|\leq \epsilon'$,
there exists $v_\epsilon\in H_{\widetilde{\Gamma},\alpha}$
such that
$\|v_\epsilon\|=1$,
$\supp v_\epsilon\subset U$, and
$\|(H_{\widetilde{\Gamma},\alpha}-E)v_\epsilon\|\leq 2\epsilon$.

 \item 
There exists $\epsilon''>0$ satisfying the following property;
for any $\widetilde{\alpha}=(\widetilde{\alpha}_j)_{j=1}^n$ with
$|\alpha_j-\widetilde{\alpha}_j|\leq \epsilon''$,
there exists $v_\epsilon\in H_{\Gamma,\widetilde{\alpha}}$
such that
$\|v_\epsilon\|=1$,
$\supp v_\epsilon\subset U$, and
$\|(H_{\Gamma,\widetilde{\alpha}}-E)v_\epsilon\|\leq 2\epsilon$.
Moreover, $\epsilon''$ can be taken uniformly with respect to $\Gamma$
so that $\delta=\delta(\Gamma)$ is bounded uniformly from below.
\end{enumerate}
\end{lemma}
\begin{proof}
 (i)
Let $\eta\in C_0^\infty(\mathbb{R}^d)$ such that
$0\leq \eta \leq 1$, 
$\eta(x)=1$ for $x\leq \delta/4$, and 
$\eta(x)=0$ for $x\geq \delta/3$.
Let $\widetilde{\Gamma}=\{\widetilde{\gamma}_j\}_{j=1}^n$
with
$|\gamma_j-\widetilde{\gamma}_j|\leq \epsilon'$
for sufficiently small $\epsilon'$ (specified later).
Consider the map
\begin{align}
\label{3-2-01}
\Phi(x) = x + \sum_{j=1}^n \eta (x-\gamma_j)\cdot
(\widetilde{\gamma}_j-\gamma_j).
\end{align}
By definition,
$\Phi$ is a $C^\infty$ map from $\mathbb{R}^d$ to itself,
$\Phi(\gamma_j)=\widetilde{\gamma}_j$,
and
\begin{align*}
|\Phi(x)-x|
+
 |\nabla (\Phi(x)-x)|
+
 |\Delta (\Phi(x)-x)|
\leq C\epsilon'
\end{align*}
for some positive constant $C$.
Thus, by Hadamard's global inverse function theorem,
$\Phi$ is a diffeomorphism from $\mathbb{R}^d$ to itself,
for sufficiently small $\epsilon'$.

Put
$w_\epsilon := u_\epsilon\circ \Phi^{-1}$.
We can easily check $w_\epsilon \in D(H_{\widetilde{\Gamma},\alpha})$,
since the map $\Phi$ is just a translation in $B_{\delta/4}(\gamma_j)$.
We use the the coordinate change $x=\Phi(y)$ or $y=\Phi^{-1}(x)$.
By (\ref{3-2-01}) and the inverse function theorem,
we have estimates
\begin{align}
&\frac{\partial x_j}{\partial y_k}(y) = \delta_{jk} + O(\epsilon'),\quad
\frac{\partial y_j}{\partial x_k}(x) = \delta_{jk} + O(\epsilon'),\quad
\frac{\partial^2 y_j}{\partial x_k \partial x_\ell}(x) = O(\epsilon'),\notag\\
&\det\left(\frac{\partial x}{\partial y}\right)=1 + O(\epsilon')
\label{3-2-02}
\end{align}
as $\epsilon'\to 0$, 
where $\delta_{jk}$ is Kronecker's delta,
and ${\partial x}/{\partial y}=(\partial x_j/\partial y_k)_{jk}$ 
is the Jacobian matrix.
The remainder terms are uniform with respect to $x$ (or $y$),
and are equal to $0$ for 
$x \not\in \bigcup_j \bigl(B_{\delta/3}(\gamma_j)\setminus B_{\delta/4}(\gamma_j)\bigr)$.
Thus we have by (\ref{3-2-02})
\begin{align*}
& \|w_\epsilon\|^2
=\int_{\mathbb{R}^d} |u_\epsilon(y)|^2 dx
=\int_{\mathbb{R}^d} |u_\epsilon(y)|^2 
\left|\det\left(\frac{\partial x}{\partial y}\right)\right|dy
= 1+O(\epsilon').
\end{align*}
Next, by the chain rule
\begin{align*}
\Delta w_\epsilon(x)
=&
\sum_{j=1}^d \frac{\partial^2}{\partial x_j^2} u_\epsilon(y)\\
=&
\sum_{j=1}^d \frac{\partial}{\partial x_j}
\left(
\sum_{k=1}^d \frac{\partial u_\epsilon}{\partial y_k}(y)\cdot 
\frac{\partial y_k}{\partial x_j}(x)
\right)\\
=&
\sum_{j=1}^d 
\sum_{k=1}^d
\left(
\sum_{\ell=1}^d
 \frac{\partial^2 u_\epsilon}{\partial^2 y_k y_\ell}(y)
\cdot \frac{\partial y_\ell}{\partial x_j}(x)
\cdot 
\frac{\partial y_k}{\partial x_j}(x)
+
\frac{\partial u_\epsilon}{\partial y_k}(y)\cdot 
\frac{\partial^2 y_k}{\partial x_j^2}(x)
\right).
\end{align*}
Thus we have by (\ref{3-2-02})
\begin{align*}
&\|(H_{\widetilde{\Gamma},\alpha} - E)w_\epsilon\|^2\\
&\int_{\mathbb{R}^d}|(-\Delta_x-E)w_\epsilon(x)|^2 dx\\
=&\int_{\mathbb{R}^d}|(-\Delta_y-E)u_\epsilon(y)|^2 dy\cdot (1+O(\epsilon'))
+\sum_{j=1}^n
\|u_\epsilon\|_{H^2(B_{\delta/3}(\gamma_j)\setminus B_{\delta/4}(\gamma_j))}^2
\cdot O(\epsilon').\\
\leq&
\epsilon^2  (1+O(\epsilon'))
+\sum_{j=1}^n
\|u_\epsilon\|_{H^2(B_{\delta/3}(\gamma_j)\setminus B_{\delta/4}(\gamma_j))}^2
\cdot O(\epsilon').
\end{align*}
By the elliptic inner regularity estimate (Corollary \ref{corollary_agmon})
\begin{align*}
& 
\sum_{j=1}^n
\|u_\epsilon\|_{H^2(B_{\delta/3}(\gamma_j)\setminus B_{\delta/4}(\gamma_j))}^2\\
\leq&
C
\sum_{j=1}^n
\left(
\|(-\Delta u_\epsilon-E)u_\epsilon\|_{L^2(B_{\delta/2}(\gamma_j))}^2
+
\|u_\epsilon\|_{L^2(B_{\delta/2}(\gamma_j))}^2
\right)\\
\leq&
C(\epsilon^2+1),
\end{align*}
where $C$ is a positive constant independent of $u_\epsilon$.
Taking $\epsilon'$ sufficiently small
and putting $v_\epsilon = w_\epsilon/\|w_\epsilon\|$,
we conclude $v_\epsilon$ has the desired property.

(ii) We give the proof only in the case $d=2$
(the case $d=1,3$ can be treated similarly).
Let $\phi_j=\phi_{\gamma_j}$ and $\psi_j=\psi_{\gamma_j}$
be the functions introduced in the proof of Proposition \ref{proposition_green}.
Then, the function $u_\epsilon$ can be uniquely expressed as
\begin{align*}
 u_\epsilon
=
\sum_{j=1}^d C_j(\phi_j-2\pi\alpha_j\psi_j)+\widetilde{u}_\epsilon,
\end{align*}
where $C_j$ is a constant
and $\widetilde{u}_\epsilon\in H^2(\mathbb{R}^2)$ such that
$\widetilde{u}_\epsilon(\gamma_j)=0$ for every $j$.

Suppose $|\widetilde{\alpha}_j-\alpha_j|\leq \epsilon''$ 
for sufficiently small $\epsilon''$, and put
\begin{align*}
 w_\epsilon
=
\sum_{j=1}^d C_j(\phi_j-2\pi\widetilde{\alpha_j}\psi_j)+\widetilde{u}_\epsilon,
\end{align*}
and $v_\epsilon=w_\epsilon/\|w_\epsilon\|$.
Then we can prove that $v_\epsilon$ has the desired property.
\end{proof}

\begin{proposition}
\label{proposition_admissible}
Let $d=1,2,3$,
and $\Gamma_\omega$ and $\alpha_\omega$ satisfy
Assumption \ref{assumption_PA}.
Then, for $H_\omega=H_{\Gamma_\omega,\alpha_\omega}$,
\begin{align}
\label{3-2-03}
\sigma(H_\omega) 
= 
\overline{\bigcup_{(\Gamma,\alpha)\in \mathcal{A}_F}\sigma(H_{\Gamma,\alpha})}
=
\overline{\bigcup_{(\Gamma,\alpha)\in \mathcal{A}_P}\sigma(H_{\Gamma,\alpha})}
\end{align}
holds almost surely.
\end{proposition}
\begin{proof}
First, let
\begin{align*}
 \Sigma=\overline{\bigcup_{(\Gamma,\alpha)\in \mathcal{A}_F}\sigma(H_{\Gamma,\alpha})},
\end{align*}
and prove $\sigma(H_\omega)=\Sigma$ holds almost surely.

Recall that $\Gamma_\omega$ is a locally finite discrete subset
satisfying Assumption \ref{assumption_no_percolation} (so $H_\omega$ is self-adjoint), almost surely.
For such $\omega$, let $E\in \sigma(H_\omega)$.
Then, by Proposition \ref{proposition_dense},
for any $\epsilon>0$
there exists $u_\epsilon \in D(H_\omega)$ such that 
$\supp u_\epsilon$ is bounded,
$\|u_\epsilon\|=1$,
and $\|(H_\omega-E)u_\epsilon\|\leq \epsilon$.
Let $\widetilde{\Gamma}=\Gamma_\omega\cap \supp u_\epsilon$
and $\widetilde{\alpha}=(\alpha_{\omega,\gamma})_{\gamma\in \widetilde{\Gamma}}$.
Then, 
$(\widetilde{\Gamma},\widetilde{\alpha})\in \mathcal{A}_F$,
$u_\epsilon \in D(H_{\widetilde{\Gamma},\widetilde{\alpha}})$
and $\|(H_{\widetilde{\Gamma},\widetilde{\alpha}}-E)u_\epsilon\|\leq \epsilon$.
This implies $\dist(E,\Sigma)\leq \epsilon$ for any $\epsilon>0$,
so $E\in \Sigma$.
Thus we conclude $\sigma(H_\omega)\subset \Sigma$ almost surely.

Conversely,
let $E\in \sigma(H_{\Gamma,\alpha})$ for some $(\Gamma,\alpha)\in \mathcal{A}_F$.
Then, for any $\epsilon>0$,
there exists $u_\epsilon \in D(H_{\Gamma,\alpha})$ 
such that 
$\supp u_\epsilon$ is contained in some bounded open set $U$,
$\|u_\epsilon\|=1$,
and $\|(H_{\Gamma,\alpha}-E)u_\epsilon\|\leq \epsilon$.
We write
$\widetilde{\Gamma}:=\Gamma\cap U=\{\gamma_j\}_{j=1}^n$
and $\alpha_j=\alpha_{\gamma_j}$.
By the ergodicity of $(\Gamma_\omega,\alpha_\omega)$,
for any $\epsilon', \epsilon''>0$
we can almost surely find
$y\in \mathbb{R}^d$ such that
$\Gamma_{\epsilon'} := \Gamma_\omega \cap (y+U)=\{\gamma_j'\}_{j=1}^n$,
$\gamma_j'=\gamma_j+ y+\epsilon_j'$ with $|\epsilon_j'|\leq \epsilon'$,
and
$\alpha_{\omega,\gamma_j'}=\alpha_j+ \epsilon_j''$ with
$|\epsilon_j''|\leq \epsilon''$.
Taking $\epsilon'$ and $\epsilon''$ sufficiently small
and applying Lemma \ref{lemma_approximatingEF},
we can almost surely find $v_\epsilon\in D(H_\omega)$
such that
$\supp v_\epsilon$ is bounded,
$\|v_\epsilon\|=1$,
and $\|(H_\omega-E)v_\epsilon\|\leq 4\epsilon$.
Then we have $\dist(\sigma(H_\omega),E)\leq 4\epsilon$ for any $\epsilon>0$,
so $E\in \sigma(H_\omega)$ almost surely.
Thus $\Sigma\subset\sigma(H_\omega)$,
and the first equality in (\ref{3-2-03}) holds.

The proof of the second equality in (\ref{3-2-03}) is similar;
we have only to replace 
$\mathcal{A}_F$ by $\mathcal{A}_P$,
and $(\widetilde{\Gamma},\widetilde{\alpha})$
in the first part of the proof by its periodic extension.
\end{proof}

\subsection{Calculation of the spectrum}
By Proposition \ref{proposition_admissible},
the proof of Theorem \ref{theorem_main2} is reduced
to the calculation of the spectrum of $H_{\Gamma,\alpha}$
for $(\Gamma,\alpha)\in \mathcal{A}_F$ or $\mathcal{A}_P$.

First we consider the case $d=1$ and the interactions
are non-negative.
\begin{lemma}
\label{lemma_spec1}
 Let $d=1$.
Let $\Gamma$ be a finite set and $\alpha=(\alpha_\gamma)_{\gamma\in \Gamma}$
with $\alpha_\gamma\geq 0$ for every $\gamma\in \Gamma$.
Then, $\sigma(H_{\Gamma,\alpha})= [0,\infty)$.
\end{lemma}
\begin{proof}
Under the assumption of the lemma, we have
\begin{align*}
 (u, H_{\Gamma,\alpha} u)=\|\nabla u\|^2 + \sum_{\gamma\in \Gamma}\alpha_\gamma |u(\gamma)|^2\geq 0
\end{align*}
for any $u\in D(H_{\Gamma,\alpha})$.
Thus $\sigma(H_{\Gamma,\alpha})\subset [0,\infty)$.
The inverse inclusion $\sigma(H_{\Gamma,\alpha})\supset [0,\infty)$
follows from \cite[Theorem II-2.1.3]{Al-Ge-Ho-Ho}.
\end{proof}
Lemma \ref{lemma_spec1} seems obvious,
but the same statement does \textit{not} hold when $d=2,3$,
since the point interaction is always negative in that case,
as stated in the introduction.

Next we consider the other cases.
In the following lemmas, 
the sequence $\alpha=(\alpha_\gamma)_{\gamma\in\Gamma}$ is assumed to be a 
\textit{constant sequence},
that is,  all the coupling constants $\alpha_\gamma$ are the same. 
We denote the common coupling constant $\alpha_\gamma$ also by $\alpha$,
by abuse of notation.
\begin{lemma}
\label{lemma_spec2}
Let $d=1$, and 
$x_1,\ldots,x_N$ be $N$ distinct points in $\mathbb{R}$
with $2\leq N<\infty$.
For $L>0$,
put $\Gamma_{N,L}=\{L x_j\}_{j=1}^N$.
Let $\alpha$
be a constant sequence on $\Gamma_{N,L}$
with common coupling constant $\alpha<0$.
Then, the following holds.

\begin{enumerate}
 \item Let $N=2$ and $|x_1-x_2|=1$.
Then, $H_{\Gamma_{2,L},\alpha}$ has only one negative eigenvalue $E_1(L)$
for $L \leq -2/\alpha$, and two negative eigenvalues 
$E_1(L)$ and $E_2(L)$ ($E_1(L)<E_2(L)$) for $L>-2/\alpha$.
The function $E_1(L)$ (resp.\ $E_2(L)$) 
is continuous and monotone increasing (resp.\ decreasing)
with respect to $L\in (0,\infty)$ (resp.\ $L\in (-2/\alpha,\infty)$), 
and
\begin{align*}
 \lim_{L\to +0} E_1(L)=-\alpha^2,\quad
 \lim_{L\to \infty} E_1(L)=-\frac{\alpha^2}{4},\\
 \lim_{L\to -2/\alpha+0} E_2(L)=0,\quad
 \lim_{L\to \infty} E_2(L)=-\frac{\alpha^2}{4}.
\end{align*}

 \item Let $N\geq 3$.
Then the operator $H_{\Gamma_{N,L},\alpha}$ has at least 
one negative eigenvalue for any $L>0$.
The smallest eigenvalue
$E_1(L)$ is simple, continuous and monotone increasing with respect to
$L\in (0,\infty)$, and
\begin{align*}
 \lim_{L\to +0} E_1(L)=-\frac{(N\alpha)^2}{4},\quad
 \lim_{L\to \infty} E_1(L)=-\frac{\alpha^2}{4}.
\end{align*}

\end{enumerate}
\end{lemma}
\begin{proof}
According to \cite[Theorem II-2.1.3]{Al-Ge-Ho-Ho},
$H_{\Gamma_{N,L},\alpha}$ has a negative eigenvalue $E=-s^2$ ($s>0$)
if and only if $\det M=0$,
where $M=(M_{jk})$ is the $N\times N$ matrix given by
\begin{align*}
 M_{jk}= 
\begin{cases}
-\alpha^{-1}- (2s)^{-1} & (j=k),\\
-(2s)^{-1}e^{-sL |x_j-x_k|} & (j\not=k).
\end{cases}
\end{align*}
Let $\widetilde{M}=(\widetilde{M}_{jk})$ be the $N\times N$-matrix given by
$\widetilde{M}_{jk}=e^{-sL|x_j-x_k|}$. Then, 
since $M= -(2s)^{-1}(2s/\alpha\cdot I +\widetilde{M})$
($I$ is the identity matrix),
\begin{align*}
\det M=0 &\Leftrightarrow \mbox{$M$ has eigenvalue $0$}\\
&\Leftrightarrow
\mbox{$-2s/\alpha$ coincides with one of eigenvalues of $\widetilde{M}$}.
\end{align*}

(i)
Let $N=2$ and $|x_1-x_2|=1$.
Then the eigenvalues of $\widetilde{M}$ are $1\pm e^{-sL}$.
So we have $E_1(L)=-s_1(L)^2$ and $E_2(L)=-s_2(L)^2$,
where $s=s_1(L)$ and $s=s_2(L)$ are solutions of
\begin{align}
\label{3-3-01}
 -\frac{2 s}{\alpha} = 1 + e^{-sL},\quad
 -\frac{2 s}{\alpha} = 1 - e^{-sL},
\end{align}
respectively, if the solutions exist.
Then the statement can be proved by inspecting the graphs of both sides
of (\ref{3-3-01})
(see Figure \ref{figure_d1_1}, \ref{figure_d1_2}).
\begin{figure}[htbp]
 \begin{center}
  \begin{tabular}{cc}
\begin{minipage}[c]{6cm}
\includegraphics[width=6cm]{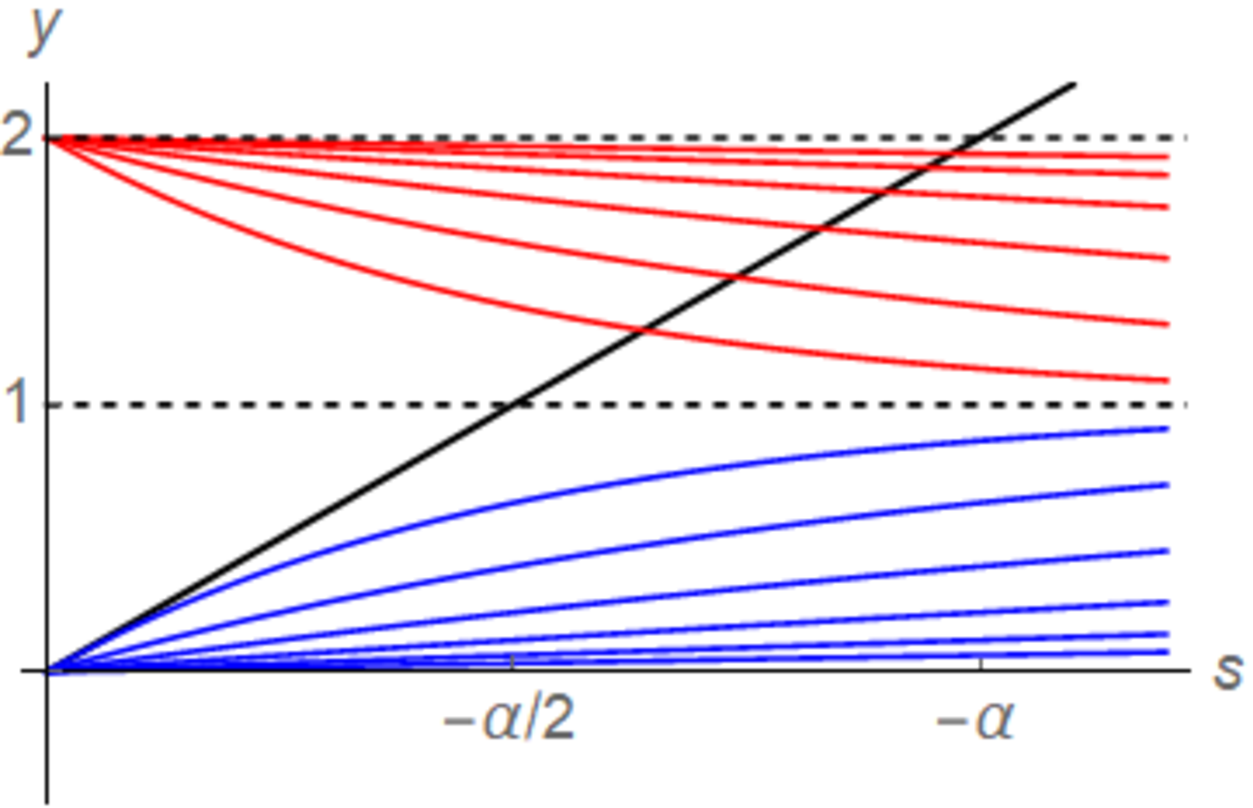} 
\caption{Graphs of both sides of (\ref{3-3-01}) for $\alpha=-1$ and $L=2^n$ ($n=-4,\ldots,1$).}
\label{figure_d1_1}
\end{minipage}
&
\begin{minipage}[c]{6cm}
\includegraphics[width=6cm]{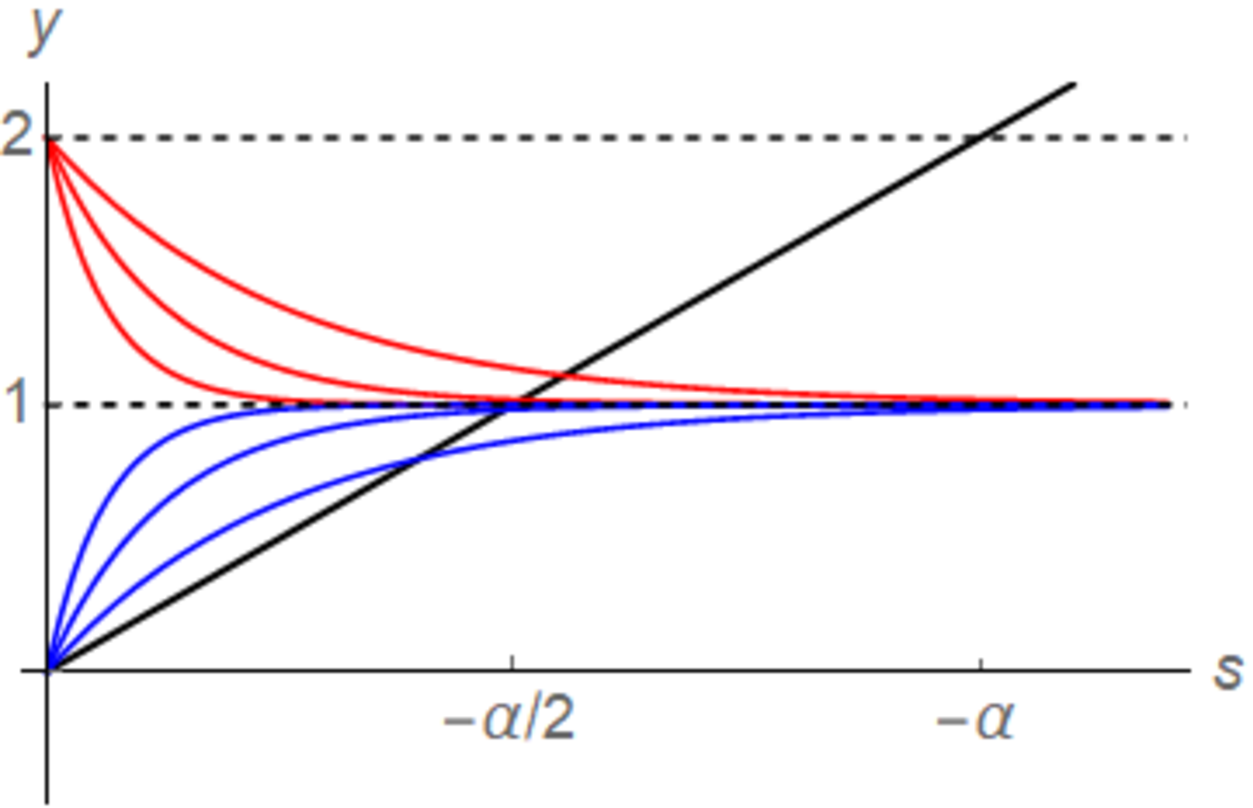} 
\caption{Graphs of both sides of (\ref{3-3-01}) for $\alpha=-1$ and $L=2^n$ ($n=2,\ldots,4$).}
\label{figure_d1_2}
\end{minipage}

 \\
  \end{tabular}
 \end{center}
\end{figure}

(ii) Let $N\geq 3$.
Let $\mu_1(s,L)$ be the largest eigenvalue of $\widetilde{M}$.
Since $\widetilde{M}$ is a symmetric matrix with positive components,
we can prove the following properties by the Perron--Frobenius theorem 
and the min-max principle.
\begin{itemize}
 \item The eigenvalue $\mu_1(s,L)$ is simple and positive,
and there is an eigenvector with only positive components.

 \item $\mu_1(s,L)$ is continuous and strictly monotone decreasing
with respect to $sL \in (0,\infty)$.

\item For fixed $L>0$,
$\displaystyle\lim_{s\to 0}\mu_1(s,L)=N$, 
$\displaystyle\lim_{s\to \infty}\mu_1(s,L)=1$.
The same properties also hold if we replace $s$ and $L$.
\end{itemize}
\begin{figure}[htbp]
 \begin{center}
  \begin{tabular}{cc}
\begin{minipage}[c]{6cm}
\includegraphics[width=6cm]{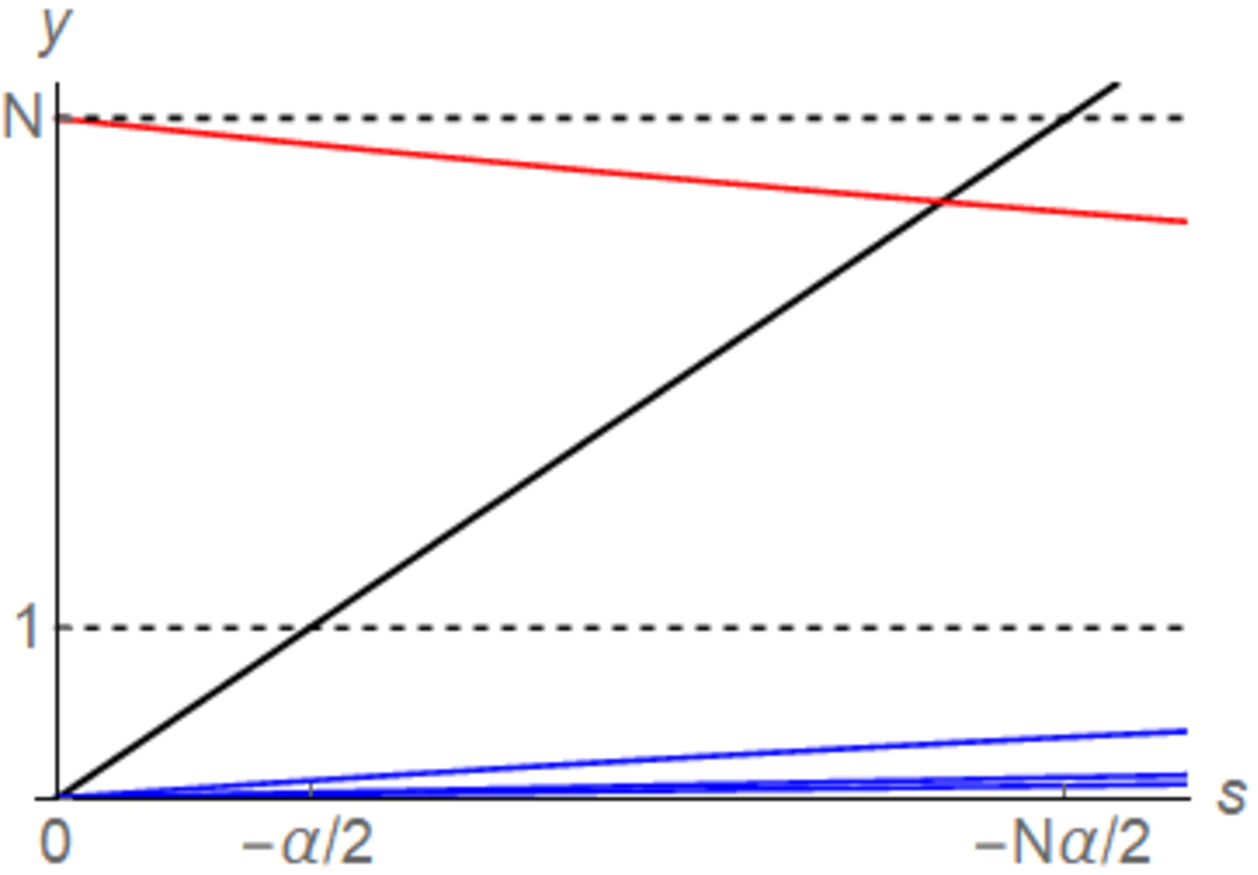} 
\caption{Graphs of $-2 s/\alpha$ and eigenvalues of $\widetilde{M}$ 
for $N=4$, $\alpha=-1$ and $L=1/16$.}
\label{figure_d1_3}
\end{minipage}
&
\begin{minipage}[c]{6cm}
\includegraphics[width=6cm]{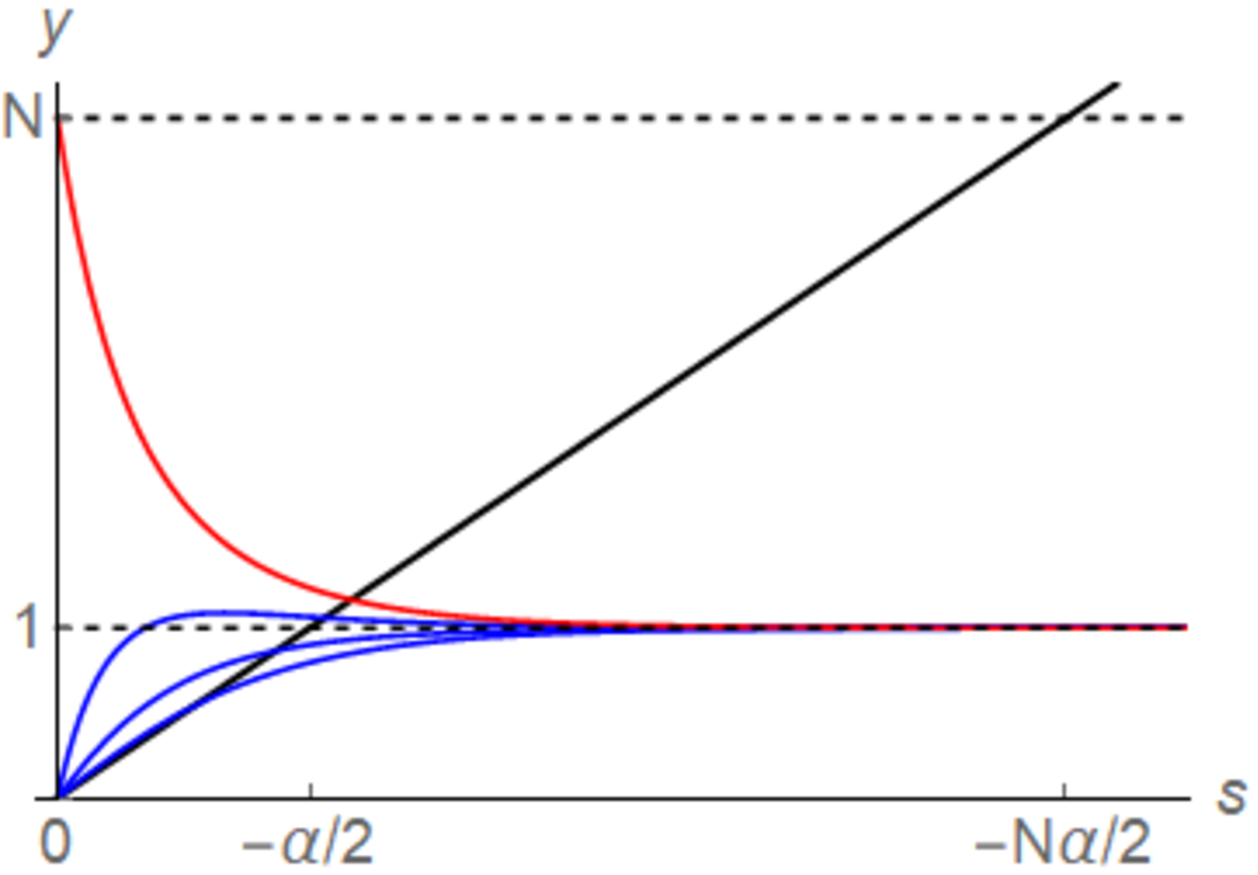} 
\caption{Graphs of $-2 s/\alpha$ and eigenvalues of $\widetilde{M}$ 
for $N=4$, $\alpha=-1$ and $L=4$.}
\label{figure_d1_4}
\end{minipage}

 \\
  \end{tabular}
 \end{center}
\end{figure}
\noindent
In Figure \ref{figure_d1_3}, \ref{figure_d1_4},
we give the graphs of $-\alpha s /2$ and eigenvalues of $\widetilde{M}$
for $N=4$, $x_j=j$ ($j=1,\ldots,4$), $\alpha=-1$ and $L=1/16,4$.

By the above properties and $\alpha<0$,
there exists a unique positive solution $s=s_1(L)$ of 
the equation $-2s/\alpha=\mu_1(s,L)$. The function $s_1(L)$ is continuous 
and strictly monotone decreasing on $(0,\infty)$.
Moreover, by inspecting the limiting equation $-2s/\alpha=\mu_1(s,0)=N$ and
$-2s/\alpha=\mu_1(s,\infty)=1$, we see that
\begin{align*}
 \lim_{L\to 0}s_1(L)=-\frac{N\alpha}{2},\quad
 \lim_{L\to \infty}s_1(L)=-\frac{\alpha}{2}.
\end{align*}
Since $E_1(L)=-s_1(L)^2$, the statement holds.

\end{proof}

\begin{lemma}
\label{lemma_spec3}
Let $d=2$.
For $L>0$,
let $\Gamma_L=\{\gamma_1, \gamma_2\}$ with $|\gamma_1-\gamma_2|=L$.
Let $\alpha$ be a constant sequence on $\Gamma_L$
with common coupling constant $\alpha\in \mathbb{R}$.
Then, $H_{\Gamma_L,\alpha}$ has only one negative eigenvalue $E_1(L)$ for 
$L\leq e^{2\pi \alpha}$,
and two negative eigenvalues 
$E_1(L)$ and $E_2(L)$ ($E_1(L)<E_2(L)$) for $L> e^{2\pi \alpha}$.
The function $E_1(L)$ (resp.\ $E_2(L)$) is continuous, 
monotone increasing (resp.\ decreasing)
with respect to $L\in (0,\infty)$ (resp.\ $L\in (e^{2\pi \alpha},\infty)$), 
and
\begin{align*}
\lim_{L\to +0}E_1(L)=-\infty,\quad
\lim_{L\to \infty}E_1(L)=-4e^{-4\pi\alpha-2\gamma_E},\\
\lim_{L\to e^{2\pi \alpha}+0}E_2(L)=0,\quad
\lim_{L\to \infty}E_2(L)=-4e^{-4\pi\alpha-2\gamma_E},
\end{align*}
where $\gamma_E$ is the Euler constant.
\end{lemma}
\begin{proof}
By \cite[Theorem II-4.2]{Al-Ge-Ho-Ho},
$H_{\Gamma_L,\alpha}$ has a negative eigenvalue
$E=-s^2$ $(s>0)$
if and only if $\det M=0$, where $M=(M_{jk})$ is a $2\times 2$-matrix given by
\begin{align*}
 M_{jk}=
\begin{cases}
 (2\pi)^{-1}(2\pi \alpha + \gamma_E + \log(s/2)) & (j=k),\\
 - \frac{i}{4}H_0^{(1)}(isL) & (j\not= k).
\end{cases}
\end{align*}
Here $H_0^{(1)}$ is the $0$-th order Hankel function of the first kind.
By \cite[9.6.4]{Ab-St}, we have
\begin{align*}
  - \frac{i}{4}H_0^{(1)}(isL) = -\frac{1}{2\pi}K_0(sL),
\end{align*}
where $K_\nu(z)$ is the $\nu$-th order modified Bessel function
of the second kind.
Thus $\det M=0$ if and only if one of the following two equations hold.
\begin{align}
\label{3-3-02}
 f(s,L):= 2\pi \alpha + \gamma_E + \log\frac{s}{2} - K_0(sL)=0,\\
\label{3-3-03}
 g(s,L):= 2\pi \alpha + \gamma_E + \log\frac{s}{2} + K_0(sL)=0.
\end{align}
Let us review formulas for the modified Bessel functions
\cite[9.6.23,9.6.27,9.6.13, 9.7.2]{Ab-St}.
\begin{align}
 \label{3-3-04}
K_\nu(z)=\frac{\pi^{1/2}}{\Gamma(\nu+1/2)}\left(\frac{z}{2}\right)^{\nu}
\int_1^\infty e^{-zt}(t^2-1)^{\nu-1/2}dt
\quad (|\arg z|<\frac{\pi}{2}),
\end{align}
\begin{align}
 K_0'(z) 
= -K_1(z),\label{3-3-05}
\end{align}
\begin{align}
\label{3-3-06}
 K_0(z)= -\log\frac{z}{2}-\gamma_E + O(z^2\log z) \quad \mbox{as } z\to 0,\\
\label{3-3-07}
 K_0(z)= \sqrt{\frac{\pi}{2z}}e^{-z}(1+O(z^{-1}))\quad \quad \mbox{as } z\to +\infty.
\end{align}
By (\ref{3-3-04})-(\ref{3-3-07}), 
we see that $K_\nu(z)>0$ for $z>0$ and $\nu > -1/2$,
and
\begin{align*}
&\frac{\partial f}{\partial s}=\frac{1}{s}+LK_1(sL)>0,\quad
\frac{\partial f}{\partial L}= s K_1(sL)>0,\\
 &\lim_{s\to +0}f(s,L)=-\infty,\quad
 \lim_{s\to \infty}f(s,L)=\infty,\\
& \lim_{L\to +0}f(s,L)=-\infty,\quad
\lim_{L\to \infty}f(s,L)=2\pi\alpha + \gamma_E +\log\frac{s}{2}.
\end{align*}
The graphs of $y=f(s,L)$ are given as curves below the dashed curve
in 
Figure \ref{figure_d2_1}, \ref{figure_d2_2}.
Here the dashed curve is the limiting curve $y=2\pi\alpha+\gamma_E+\log s/2$.
\begin{figure}[htbp]
 \begin{center}
  \begin{tabular}{cc}
\begin{minipage}[c]{6cm}
\includegraphics[width=6cm]{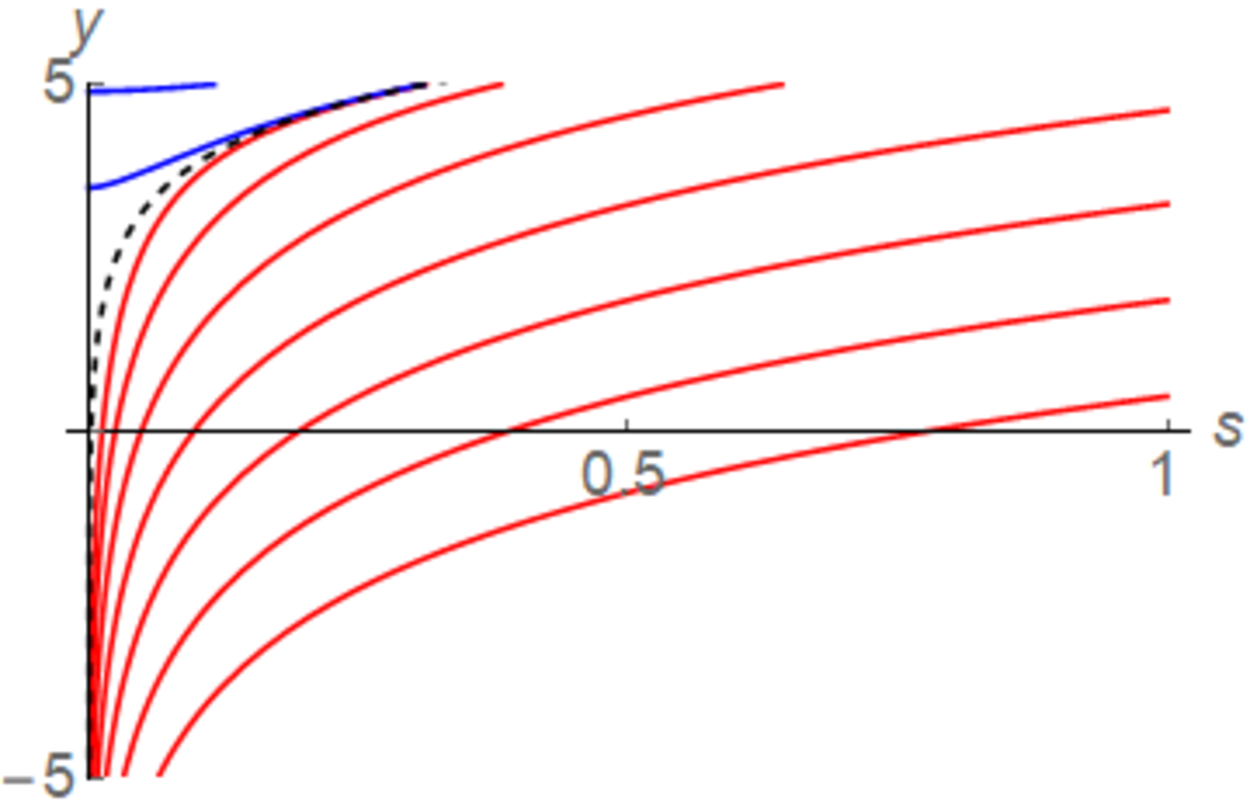} 
\caption{Graphs of $y=f(s,L)$ and $y=g(s,L)$
for $\alpha=1$ and $L=4^n$ ($n=-4,\ldots,2$).}
\label{figure_d2_1}
\end{minipage}
&
\begin{minipage}[c]{6cm}
\includegraphics[width=6cm]{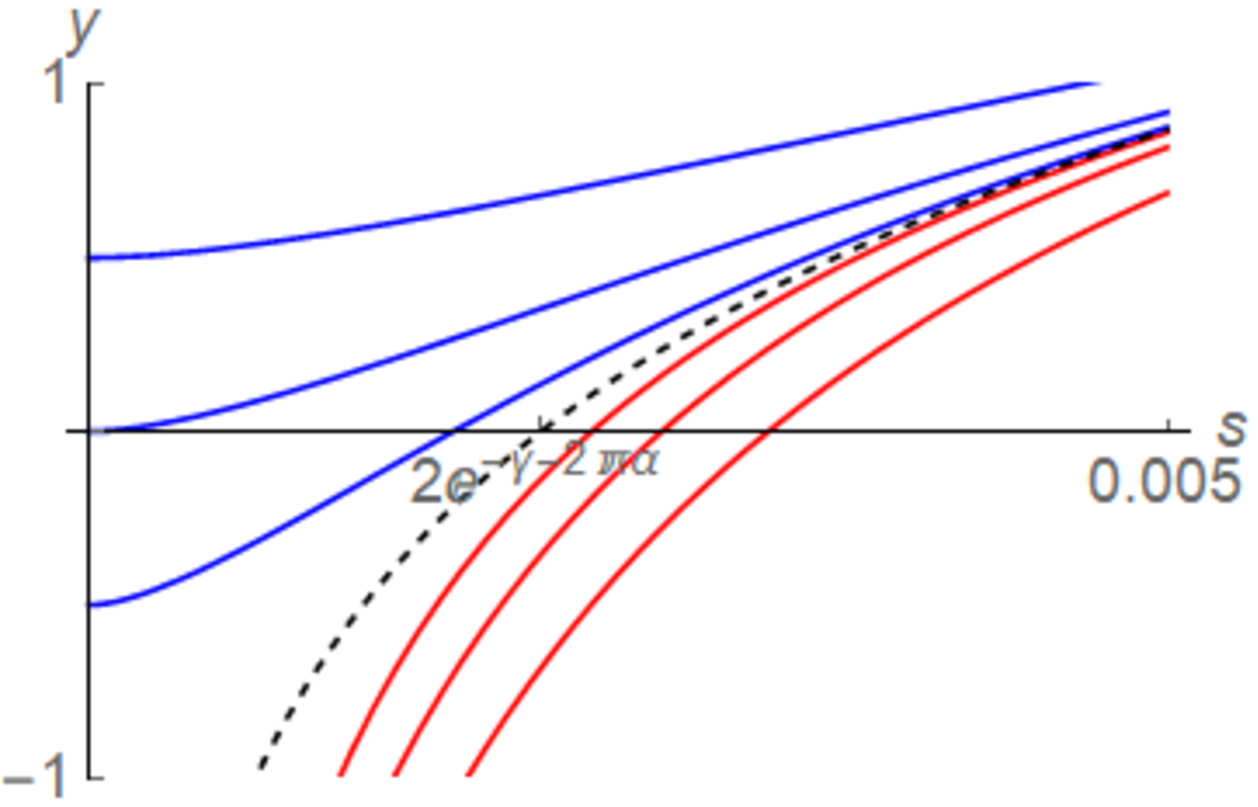} 
\caption{Graphs of $y=f(s,L)$ and $y=g(s,L)$
for $\alpha=1$ and $L=e^{2\pi\alpha+n/2}$ ($n=-1,0,1$).}
\label{figure_d2_2}
\end{minipage}

 \\
  \end{tabular}
 \end{center}
\end{figure}

By these properties, we conclude
that the equation (\ref{3-3-02}) 
has unique positive solution $s=s_1(L)$ for any $L>0$, and
\begin{align*}
\displaystyle\lim_{L \to +0}s_1(L)=\infty, \quad
\displaystyle\lim_{L \to \infty}s_1(L)=2e^{-2\pi\alpha-\gamma_E}.
\end{align*}

Next, again by (\ref{3-3-04})-(\ref{3-3-07}),
\begin{align*}
\frac{\partial g}{\partial s}&=\frac{1}{s}-LK_1(sL)
=\frac{1}{s}
-sL^2\int_1^\infty e^{-sLt}(t^2-1)^{1/2}dt\\
&>\frac{1}{s}-sL^2\int_0^\infty e^{-sLt}tdt= 0,\\
 \frac{\partial g}{\partial L}&= -s K_1(sL)<0,
\end{align*}
\begin{align*}
& \lim_{s\to +0}g(s,L)=2\pi \alpha -\log L,\quad
 \lim_{s\to \infty}g(s,L)=\infty,\\
& \lim_{L\to +0}g(s,L)=\infty,\quad
\lim_{L\to \infty}g(s,L)=2\pi\alpha + \gamma_E +\log\frac{s}{2}.
\end{align*}
The graphs of $y=g(s,L)$ are given as curves above the dashed curve in 
Figure \ref{figure_d2_1}, \ref{figure_d2_2}.
By these properties, we conclude
the equation (\ref{3-3-03}) has no positive solution 
for $ L\leq e^{2\pi \alpha}$, has unique positive solution $s=s_2(L)$ for
$L>e^{2\pi \alpha}$, and
\begin{align*}
\lim_{L \to e^{2\pi \alpha}+0}s_2(L)=0, \quad
\lim_{L \to \infty}s_2(L)=2e^{-2\pi\alpha-\gamma_E}.
\end{align*}
Since $E_1(L)=-s_1(L)^2$ and $E_2(L)=-s_2(L)^2$, 
the statements hold.
\end{proof}

\begin{lemma}
\label{lemma_spec4}
Let $d=3$.
For $L>0$,
let $\Gamma_L=\{\gamma_1,\gamma_2\}$ with $|\gamma_1-\gamma_2|=L$.
Let $\alpha$ be a constant sequence on $\Gamma_L$
with common coupling constant $\alpha\in \mathbb{R}$.
Then, the following holds.

\begin{enumerate}
 \item Assume $\alpha\geq 0$.
Then, 
$H_{\Gamma_L,\alpha}$ has no negative eigenvalue for $L \geq 1/(4\pi\alpha)$,
and has one negative eigenvalue $E_1(L)$ for  $0<L < 1/(4\pi\alpha)$
(when $\alpha=0$, we interpret $1/(4\pi\alpha)=\infty$ and the first case does not occur).
The function $E_1(L)$ is continuous, monotone increasing with respect to 
$L\in (0,1/(4\pi\alpha))$, and
\begin{align*}
 \lim_{L\to +0}E_1(L)=-\infty,\quad
 \lim_{L\to 1/(4\pi\alpha)-0}E_1(L)=0.
\end{align*}
 \item Assume $\alpha<0$.
Then, 
$H_{\Gamma_L,\alpha}$ 
has one negative eigenvalue $E_1(L)$ for  $L\leq 1/(-4\pi\alpha)$,
and two negative eigenvalues $E_1(L)$ and $E_2(L)$ ($E_1(L)<E_2(L)$)
for $L> 1/(-4\pi\alpha)$.
The function $E_1(L)$ (resp.\ $E_2(L)$)
is continuous, monotone increasing (resp.\ decreasing)
with respect to $L\in (0,\infty)$ (resp.\ $L\in (1/(-4\pi\alpha),\infty)$), 
and
\begin{align*}
 \lim_{L\to +0}E_1(L)=-\infty,\quad
 \lim_{L\to \infty}E_1(L)=-(4\pi\alpha)^2,\\
 \lim_{L\to 1/(-4\pi\alpha)+0}E_2(L)=0,\quad
 \lim_{L\to \infty}E_2(L)=-(4\pi\alpha)^2.
\end{align*}
\end{enumerate}
\end{lemma}
\begin{proof}
By \cite[Theorem II-1.1.4]{Al-Ge-Ho-Ho},
$H_{\Gamma_L,\alpha}$ has a negative eigenvalue
$E=-s^2$ $(s>0)$
if and only if $\det M=0$, where $M=(M_{jk})$ is a $2\times 2$-matrix given by
\begin{align*}
 M_{jk} = 
\begin{cases}
\displaystyle
 \alpha +\frac{s}{4\pi} & (j=k),
\vspace{2mm}
\\
\displaystyle
-\frac{e^{-sL}}{4\pi L} & (j\not=k).
\end{cases}
\end{align*}
So $\det M=0$ if and only if one of the following equations holds.
\begin{align}
\label{3-3-08}
 4\pi \alpha + s = \frac{e^{-sL}}{L},\\
\label{3-3-09}
 4\pi \alpha + s = -\frac{e^{-sL}}{L}.
\end{align}
The graphs of both sides of (\ref{3-3-08}) and (\ref{3-3-09})
are given in Figure \ref{figure_d3_1}, \ref{figure_d3_2}.
\begin{figure}[htbp]
 \begin{center}
  \begin{tabular}{ccc}
\begin{minipage}[c]{6cm}
\includegraphics[width=6cm]{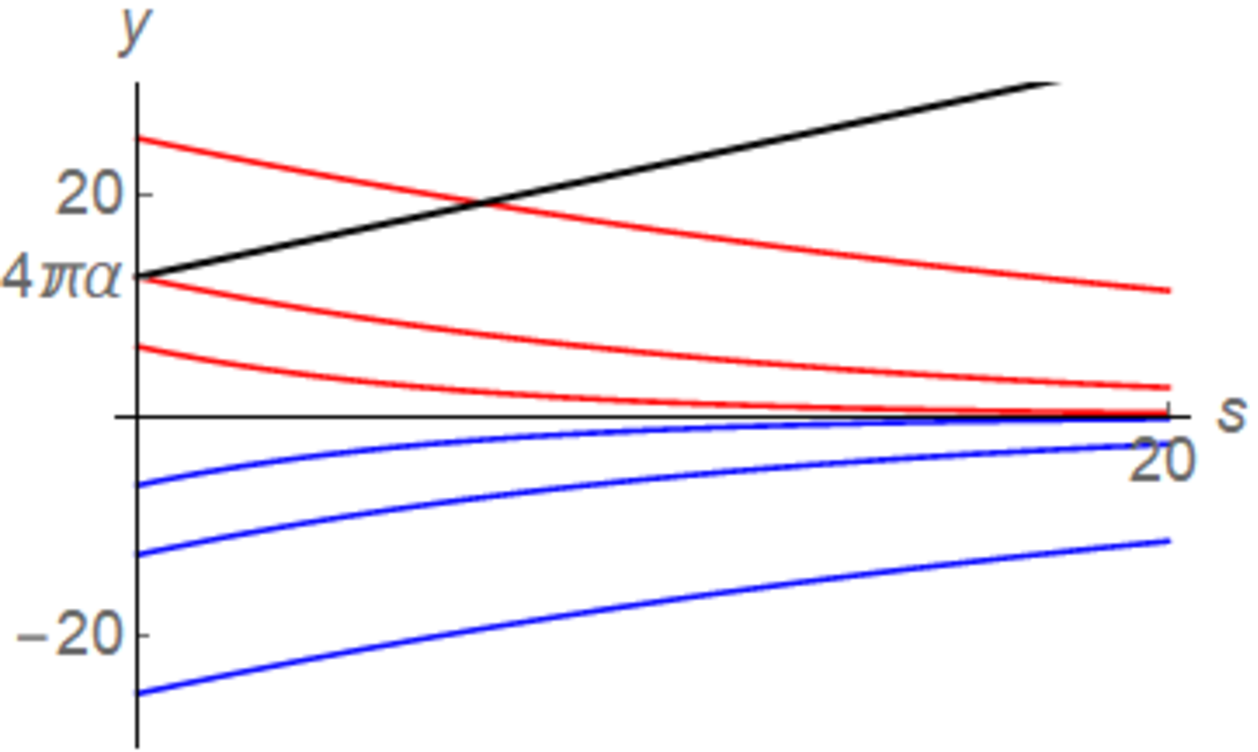} 
\caption{Graphs of both sides of (\ref{3-3-08}) and (\ref{3-3-09}) 
for $\alpha=1$ and $L=2^n/(4\pi\alpha)$ ($n=-1,0,1$).}
\label{figure_d3_1}
\end{minipage}
&
\begin{minipage}[c]{6cm}
\includegraphics[width=6cm]{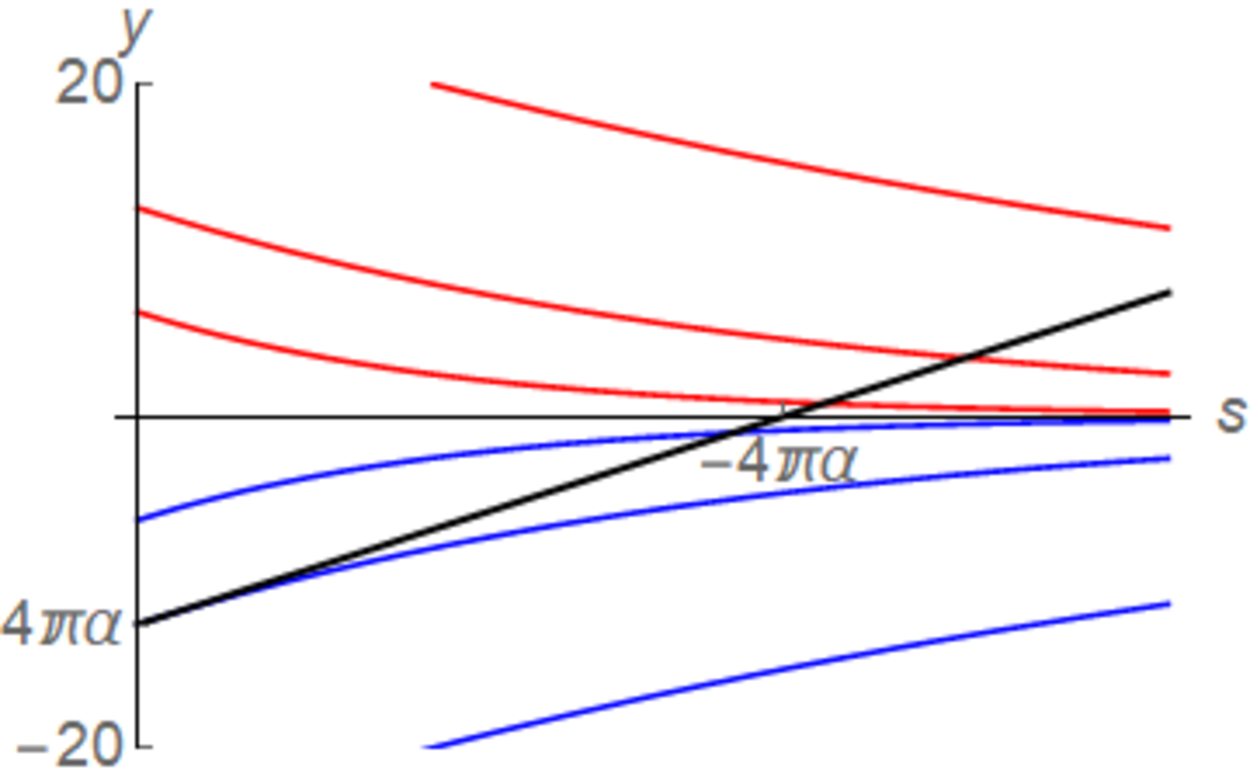} 
\caption{Graphs of both sides of (\ref{3-3-08}) and (\ref{3-3-09}) 
for $\alpha=-1$ and $L=2^n/(-4\pi\alpha)$ ($n=-1,0,1$).}
\label{figure_d3_2}
\end{minipage}
  \end{tabular}
 \end{center}
\end{figure}

By inspecting the graphs, we conclude the following.
\begin{enumerate}
 \item For $\alpha \geq 0$,
the equation (\ref{3-3-08}) has no positive solution for
 $L \geq 1/(4\pi\alpha)$, 
and has one positive solution $s=s_1(L)$ for  $0<L<1/(4\pi\alpha)$.
Moreover,
$\displaystyle \lim_{L\to +0}s_1(L)=\infty$,
$\displaystyle \lim_{L\to 1/(4\pi\alpha)-0}s_1(L)=0$.
The equation (\ref{3-3-09}) has no positive solution.

 \item For $\alpha<0$,
the equation (\ref{3-3-08}) 
has one positive solution $s=s_1(L)$ for any $L>0$,
and 
$\displaystyle \lim_{L\to +0}s_1(L)=\infty$,
$\displaystyle \lim_{L\to \infty}s_1(L)=-4\pi \alpha$.
The equation (\ref{3-3-09}) has no positive solution for $L\leq 1/(-4\pi\alpha)$,
has one positive solution $s=s_2(L)$ for $L> 1/(-4\pi\alpha)$,
and 
$\displaystyle \lim_{L\to 1/(-4\pi\alpha)+0}s_2(L)=0$,
$\displaystyle \lim_{L\to \infty}s_2(L)=-4\pi \alpha$.
\end{enumerate}
These facts and $E_1(L)=-s_1(L)^2$, $E_2(L)=-s_2(L)^2$ imply the statements.
\end{proof}
\begin{proof}[Proof of Theorem \ref{theorem_main2}]
Put
\begin{align*}
\Sigma 
= \overline{\bigcup_{(\Gamma,\alpha)\in \mathcal{A}_F}\sigma(H_{\Gamma,\alpha})}
. 
\end{align*}
By Proposition \ref{proposition_admissible},
we have $\sigma(H_\omega) = \Sigma$ almost surely.

First consider the case
$d=1$ and $\supp \nu \subset [0,\infty)$.
Then, for any $(\Gamma,\alpha)\in \mathcal{A}_F$,
we have $\sigma(H_{\Gamma,\alpha})=[0,\infty)$
by Lemma \ref{lemma_spec1}.
So $\Sigma =[0,\infty)$.

In all other cases, we have to prove $\Sigma=\mathbb{R}$.
Since $\sigma(H_{\Gamma,\alpha}) = [0,\infty)$ for
$\Gamma=\emptyset$, we have only to prove $(-\infty,0)\subset \Sigma$.

Consider the case $d=1$ and $\supp\nu\cap(-\infty,0)\not=\emptyset$.
Let $\Gamma_{N,L}$ given in Lemma \ref{lemma_spec2},
and $\alpha$ be a constant sequence on $\Gamma_{N,L}$ 
with common coupling constant $\alpha\in \supp\nu\cap (-\infty,0)$.
Then $(\Gamma_{N,L},\alpha)\in \mathcal{A}_F$ for any $N\geq 2$ and $L>0$,
so
\begin{align*}
 \Sigma \supset \bigcup_{N \geq 2, L>0} \sigma(H_{\Gamma_{N,L},\alpha}).
\end{align*}
By Lemma \ref{lemma_spec2}, the right hand side contains $(-\infty,0)$.
When $d=2,3$, the statement can be proved similarly by using 
Lemma \ref{lemma_spec3}, \ref{lemma_spec4}.
\end{proof}

In the case $d=1$ and $\supp\nu$ has negative part,
there is a simple another proof 
using the spectrum of the \textit{Kronig--Penney model} 
(see \cite{Kr-Pe,Al-Ge-Ho-Ho}).
\begin{proof}[Another proof of Theorem \ref{theorem_main2} (i)]
Put
\begin{align*}
\Sigma 
= \overline{\bigcup_{(\Gamma,\alpha)\in \mathcal{A}_P}\sigma(H_{\Gamma,\alpha})}
. 
\end{align*}
By Proposition \ref{proposition_admissible},
we have $\sigma(H_\omega) = \Sigma$ almost surely.

Assume $d=1$ and 
$(-\infty,0)\cap \supp \nu\not=\emptyset$.
It is sufficient to show $\Sigma \supset (-\infty,0)$.
For $L>0$, let $\Gamma_L=L\mathbb{Z}$, and 
$\alpha$ be a constant sequence on $\Gamma_L$
with common coupling constant $\alpha\in \supp\nu\cap(-\infty,0)$.
Then $(\Gamma_L,\alpha)\in\mathcal{A}_P$.
By \cite[Theorem III.2.3.1]{Al-Ge-Ho-Ho},
the spectrum of  $H_{\Gamma_L,\alpha}$ is given by
\begin{align*}
\sigma(H_{\Gamma_L,\alpha}) =
\left\{k^2 \in \mathbb{R} \mid 
\left|
\cos(kL) +{\alpha}/(2k)\sin(kL)\right|
\leq 1\right\}.
\end{align*}
Put $k=is$ for $s>0$.
Then, $E=-s^2\in \sigma(H_{\Gamma_L,\alpha})$ 
if and only if
\begin{align}
\label{3-3-10}
 \left|\cosh(sL)+\alpha/(2s)\sinh(sL)\right|\leq 1.
\end{align}
Take arbitrary $s_0>0$, and let $s\in (0,s_0]$.
Consider the Taylor expansion with respect to $L$
\begin{align}
\label{3-3-11}
 f(s,L):=\cosh(sL)+\frac{\alpha}{2s}\sinh(sL)
= 1 +\frac{\alpha}{2}L + O(L^2) \quad \mbox{as }L\to 0.
\end{align}
The remainder term is uniform with respect to $s\in (0,s_0]$.
Since $\alpha<0$,
(\ref{3-3-11}) implies (\ref{3-3-10}) holds for sufficiently small $L$
uniformly with respect to $s\in (0,s_0]$
(see also Figure \ref{figure_d3_3}).
Thus $[-s_0^2,0)\subset \sigma(H_{\Gamma_L,\alpha})$ for 
sufficiently small $L$, so $(-\infty,0)\subset \Sigma$.
\begin{figure}[htbp]
 \begin{center}
\includegraphics[width=6cm]{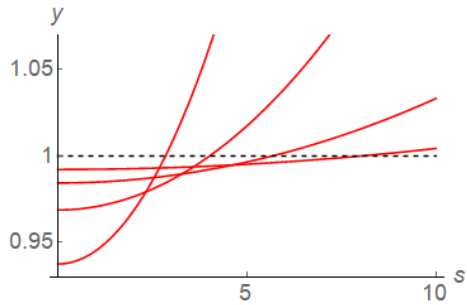} 
\caption{Graphs of $y=f(s,L)$ for $L=2^{-n}$ ($n=3,\ldots,6$).
As $L\to +0$, the negative band
becomes longer and longer.
}
\label{figure_d3_3}
 \end{center}
\end{figure}

\end{proof}

\section{Appendix}
\subsection{Elliptic inner regularity estimate}
The following is a special case of the elliptic inner regularity theorem
(\cite[Theorem 6.3]{Ag}).
\begin{theorem}
\label{theorem_agmon}
Let $U$ be an open set in $\mathbb{R}^d$
and $u\in L^2(U)$.
Assume that there exists a positive constant $M$ such that
\begin{equation}
\label{ap01}
  |(u,\Delta \phi )_{L^2(U)}|\leq M \|\phi\|_{L^2(U)}
\end{equation}
holds for every $\phi \in C_0^\infty(U)$.
Then, $u\in H^2_{\rm loc}(U)$.
Moreover, for any open set $V$ such that $\overline{V}$
is a compact subset of $U$,
there exists a positive constant 
$C$ dependent only on $U$ and $V$ such that
\begin{align*}
\|u\|_{H^2(V)}\leq C(M + \|u\|_{L^2(U)}),
\end{align*}
where $M$ is the constant in (\ref{ap01}).
\end{theorem}
From 
Theorem \ref{theorem_agmon},
we have the following corollary useful for our purpose.
\begin{corollary}
\label{corollary_agmon}
Let $U$, $V$ be open sets in $\mathbb{R}^d$ such that
$\overline{V}\subset U$ and 
\begin{align*}
\dist(\partial U, V)\geq \delta
\end{align*}
for some positive constant $\delta$.
Let $u\in L^2(U)$ such that $\Delta u\in L^2(U)$ in the distributional sense.
Then, $u\in H^2_{\rm loc}(U)$, and 
there exists a constant $C$ dependent only on $\delta$ and the dimension $d$
such that
\begin{equation}
 \label{app04}
\|u\|_{H^2(V)}^2
\leq
C \left(
\|\Delta u \|_{L^2(U)}^2+
\| u \|_{L^2(U)}^2
\right).
\end{equation}
\end{corollary}
\begin{proof}
 Put $\epsilon=\delta/(2d)$.
For $x_0\in \mathbb{R}^d$,
consider open cubes $Q=x_0 + (-\epsilon,\epsilon)^d$
and $Q'=x_0 + (-\epsilon/2,\epsilon/2)^d$.
When $Q\subset U$, we have
\begin{equation*}
 |(u, \Delta \phi)_{L^2(Q)}|
=
 |(\Delta u, \phi)_{L^2(Q)}|
\leq
\|\Delta u\|_{L^2(Q)}
\|\phi\|_{L^2(Q)}
\end{equation*}
for every $\phi\in C_0^\infty(Q)$.
Then the assumption of Theorem \ref{theorem_agmon} is satisfied
with $U=Q$, $V=Q'$, and $M=\|\Delta u\|_{L^2(Q)}$,
and we have
\begin{equation}
\label{app05}
 \|u\|_{H^2(Q')}^2\leq C(\|\Delta u\|_{L^2(Q)}^2 + \|u\|_{L^2(Q)}^2)
\end{equation}
for some positive constant $C$
dependent only on $\delta$ and dimension $d$.
We collect all the cubes $Q'$ such that
the center $x_0\in \epsilon \mathbb{Z}^d$ and
$Q'\cap V \not=\emptyset$.
Notice that $Q\subset U$ for such $Q'$.
Thus we have by (\ref{app05})
\begin{align*}
 \|u\|_{H^2(V)}^2
&\leq
\sum_{Q'}
 \|u\|_{H^2(Q')}^2\\
&\leq
C
\sum_{Q'}
\left(
 \|\Delta u\|_{L^2(Q)}^2
+
 \|u\|_{L^2(Q)}^2
\right)
\\
&\leq
2^{d}C
\left(
 \|\Delta u\|_{L^2(U)}^2
+
 \|u\|_{L^2(U)}^2
\right),
\end{align*}
where we use the fact $Q$ can overlap at most $2^d$ times.
\end{proof}

\textbf{Acknowledgments.}
The work of T.\ M.\ is partially supported by
JSPS KAKENHI Grant Number JP18K03329.
The work of F.\ N.\ is partially supported by
JSPS KAKENHI Grant Number JP26400145.

\end{document}